\documentclass[12pt]{article}
\renewcommand{\baselinestretch}{1.2}
\usepackage[pagebackref,letterpaper=true,colorlinks=true,pdfpagemode=noneh, urlcolor=blue,linkcolor=blue,citecolor=blue,pdfstartview=FitH]{hyperref}
\usepackage{stmaryrd}
\usepackage{amsmath,amsfonts}
\usepackage{graphicx}
\usepackage{color}
\usepackage[title]{appendix}
\usepackage{mathrsfs}
\usepackage{float} 

\usepackage{epsfig} 
\usepackage{epstopdf}
\usepackage{subfigure}

\usepackage{pgfplots}
\tikzset{elegant/.style={smooth,thick,samples=50,cyan}}
\tikzset{eaxis/.style={->,>=stealth}}

\usepackage{xspace}

\newif\ifblog
\newif\iftex
\blogfalse
\textrue

\newcommand{\thmref}[1]{Theorem~{\rm \ref{#1}}}
\newcommand{\lemref}[1]{Lemma~{\rm \ref{#1}}}

\newcommand{\propref}[1]{Proposition~{\rm \ref{#1}}}

\newcommand{\be}{\begin{eqnarray}}
\newcommand{\ee}{\end{eqnarray}}
\newcommand{\bee}{\begin{eqnarray*}}
\newcommand{\eee}{\end{eqnarray*}}

\def\d{{\rm d}}

\def\ds{{\rm d}s}

\def\P{{\mathbb P}}
\def\E{{\mathbb E}}

\def\R{{\mathbb R}}

\def\R{{\mathbb R}}

\def\p{{\partial}}

\newcommand{\cF}{{\cal F}}
\newcommand{\ep}{\varepsilon}
\newcommand{\al}{\alpha}

\newcommand{\nd}{\noindent}

\newcommand{\si}{\sigma}
\newcommand{\rt}{\rightarrow}

\newcommand{\fB}{{\mathfrak{B}}}
\newcommand{\fM}{{\mathfrak{N}}}
\newcommand{\fS}{{\mathfrak{S}}}

\newcommand{\argmax}{\mathop{\rm argmax}\limits}
\newcommand{\argmin}{\mathop{\rm argmin}\limits}
\newcommand{\sgn}{\mathop{\rm sgn}}

\newtheorem{theorem}{Theorem}[section]
\newtheorem{lemma}[theorem]{Lemma}

\newtheorem{proposition}[theorem]{Proposition}

\newtheorem{remark}[theorem]{Remark}

\newenvironment{proof}{\noindent{\sc Proof:}}{\strut\hfill $\Box$\medskip\\} 

\newcommand{\TQ}{Q^T}

\newcommand{\ctr}{CTT\xspace} 
\newcommand{\dtr}{DTT\xspace} 

\title{Continuous-time Markowitz's mean-variance model \\ under different borrowing and saving rates}

\author{Chonghu Guan\thanks{School of Mathematics, Jiaying University, Meizhou 514015, Guangdong, China. Email: \url{gchonghu@163.com}}
\and Xiaomin Shi\thanks{School of Statistics and Mathematics, Shandong University of Finance and Economics, Jinan 250100, China. Email: \url{shixm@mail.sdu.edu.cn}}
\and Zuo Quan Xu\thanks{Department of Applied Mathematics, The Hong Kong Polytechnic University, Kowloon, Hong Kong, China. 
Email: \url{maxu@polyu.edu.hk}}
}
\date{}

\begin{document}
\maketitle

\begin{abstract}
We study Markowitz's mean-variance portfolio selection problem in a continuous-time Black-Scholes market with different borrowing and saving rates. The associated Hamilton-Jacobi-Bellman equation is fully nonlinear.
Using a delicate partial differential equation and verification argument, the value function is proven to be $C^{3,2}$ smooth. 
It is also shown that there are a borrowing boundary and a saving barrier which divide the entire trading area into a borrowing-money region, an all-in-stock region, and a saving-money region in ascending order.
The optimal trading strategy is a mixture of continuous-time strategy (as suggested by most continuous-time models) and discontinuous-time strategy (as suggested by models with transaction costs): one should put all her wealth in the stock in the middle all-in-stock region, and continuously trade it in the other two regions in a feedback form of wealth and time. It is never optimal to short sale the stock.
Numerical examples are also presented to verify the theoretical results and to give more financial insights beyond them.

\bigskip
\nd {\bf Keywords.} Markowitz's mean-variance portfolio selection; fully nonlinear PDE; free boundary; dual transformation; different borrowing and saving rates
\bigskip\\
\nd {\bf 2010 Mathematics Subject Classification.} 35R35; 35Q93; 91G10; 91G30; 93E20.

\end{abstract}

\newpage
\section{Introduction}

\noindent

The famous work of Harry M. Markowitz \cite{Ma52,Ma59} inaugurated a new era in modern finance. Markowitz's mean-variance portfolio selection has become one of the most prominent modern finance theories since its inception.

Numerous Markowitz's models with new features have been studied in the literature. In the realm of continuous-time framework, Richardson \cite{R89} first applied martingale method to study Markowitz's mean-variance portfolio selection problem. Zhou and Li \cite{ZL00} used an embedding technique and stochastic linear-quadratic (LQ) control theory to study the problem. Li, Zhou and Lim \cite{LZL01} considered the problem with no-shorting constraints. Hu and Zhou \cite{HZ05} extended it to the case with random coefficients and cone constraints on the control variable. Czichowsky and Schweizer \cite{CS13} provided the most general solutions for cone-constrained Markowitz's problem including new effects resulting from the jumps in the price process. Lv, Wu and Yu \cite{LWY16} studied Markowitz's problem with random horizon in an incomplete market setting. Xiong and Zhou \cite{XZ07}, and Xiong, Xu and Zheng \cite{XXZ21} investigated the problem under partial information. Zhou and Yin \cite{ZY03}, Hu, Shi and Xu \cite{HSX21a,HSX21b} considered the problem under regime switching and trading constraints.

As a variant of Markowitz's problem, the mean-variance (or quadratic) hedging problem was introduced by Duffie and Richardson \cite{DR91} and Schweizer \cite{S92}. Pham \cite{P00} extended the problem to a general incomplete market with semimartingale price process. Gourieroux, Laurent and Pham \cite{GLP02} introduced Numéraire to the problem. The most general model was instigated by Cerny and Kallsen \cite{CK07}. We refer to Schweizer \cite{S10} for an overview of the topic. Usually, the mean-variance portfolio selection and hedging problems are considered either for discounted prices or in the presence of a risk-free asset. Cerny, Czichowsky, and Kallsen \cite{CCK21} provided a solution to both problems that is symmetric and allows for all assets to be risky.

Stochastic LQ control method is widely used to study mean-variance portfolio selection and hedging problems. This method is extremely powerful when dealing with problems with trading constraints and random coefficients (see, e.g. \cite{KT02,LZL01,LZ02,Y13,ZL00,ZY03}), but less powerful when dealing with problems with state constraints such as bankruptcy prohibition. The latter kind of problems is often dealt by martingale method (which usually requires complete market setting) or by partial differential equation (PDE) method (which requires Markovian market setting). For instance, Bielecki, Jin, Pliska and Zhou \cite{BJPZ} investigated continuous-time Markowitz's problem with bankruptcy prohibition. Using martingale method, they turned the dynamic stochastic control problem into a static random variable chosen problem that was eventually solved by optimization method. Li and Xu \cite{LX16} studied Markowitz's problem with both trading and bankruptcy prohibition constraints by PDE method. Their idea is first to transform the problem into an equivalent one with only bankruptcy prohibition constraint, then to solve the latter by the method of \cite{BJPZ}. Xia \cite{X05} established the relationship between Markowitz's problem and the expected utility maximization problem with non-negative marginal utility in incomplete market with bankruptcy prohibition. Hou and Xu \cite{HX16} examined the effect of intractable claims on the trading strategy in Markowitz's problem by martingale approach.

The optimal trading strategies obtained in the aforementioned papers are typically trading continuously all the time, which are not consistent with real practice most of time. Dai, Xu and Zhou \cite{DXZ10} studied Markowitz's problem with proportional transaction costs by PDE method. To solve the associated Hamilton-Jacobi-Bellman (HJB) equation, they first derived a related double-obstacle PDE problem through an intuitive argument. The solvability of the latter PDE was completed resolved by PDE method so that they can get a classical solution to the original HJB equation. They showed that the optimal tradings only happen when the stock price arrives at a selling-stock boundary or at a buying-stock boundary. This is a discontinuous-time trading (\dtr) strategy which fits the real practice better than those continuous-time trading (\ctr) strategies suggested by most existing models.

All the aforementioned papers assume that there is no difference between the borrowing rate and saving rate in the market, namely the borrowing and saving rates are the same all the time even though they may be modeled as stochastic processes. But, as is well-known, a gap between the two rates always exists, which is fairly large sometimes, in practice. As borrowing rate is often higher than saving rate, it discourages/panelizes investors to borrow money. Different from a vast amount of no-gap-market models in the literature, only a very limited number of papers studied gap-market models. In the book Karatzas and Shreve \cite{KS98}, a utility maximization in a gap-market model is studied by martingale and duality methods. For continuous price processes, the mean-variance portfolio selection problem is rather close to the utility maximization problem. Although the quadratic functional is not always increasing, it can be shown that the optimal wealth process should stay always in the domain where the quadratic functional is strictly increasing and hence behaves like a utility function (see, Delbaen and Schachermayer \cite{DS96}). Therefore, the method of Karatzas and Shreve \cite{KS98} may be applied to Markowitz's problem, but it cannot provide a description of the optimal borrowing, saving and all-in-stock regions such as their connectedness and monotonicity. Other investment and pricing problems are also studied in gap-market models. For instance, Fleming and Zariphopoulou \cite{FZ06}, Xu and Chen \cite{XC98} considered optimal investment and consumption problems; Bergman \cite{Be91}, Korn \cite{Ko92}, and Cvitanic and Karatzas \cite{CK93} studied option pricing problems; Guan \cite{G18} studied a utility maximization problem.

Fu, Lari-Lavassani and Li \cite{FLL10} is the only paper we can identify in the literature which tried to solve Markowitz's problem in a continuous-time market with different borrowing and saving rates. They constructed a piece-wise quadratic solution to the HJB equation, but did not verify if the solution is the value function of the original problem. Their constructed solution is not of $C^{3,2}$ smooth, but we will show the value function is of $C^{3,2}$ smooth in this paper, so \cite{FLL10} indeed did not get the right value function or the optimal strategy. Therefore, the problem is still open and we will fill this gap.

This paper investigates Markowitz's portfolio selection problem in a continuous-time Black-Scholes market with different borrowing and saving rates. We show that the whole trading area is divided by a borrowing boundary and a saving boundary into three ascending trading regions, corresponding to the optimal strategies of borrowing money, putting all wealth in the stock and saving money. The existence of the three trading regions was already observed for utility maximization problem by Fleming and Zariphopoulou in \cite{FZ06}. We prove that these regions are connected and ordered, so that our results can provide the following financial insights. When an investor's wealth is far from her target, she must borrow money to invest in the stock so as to maximize the chance to achieve her goal; by contrast, if her wealth is sufficiently close to her target, she does not need to invest all her wealth in the stock and should save some in the money account to reduce her risk; while in the middle all-in-stock region, she should keep all her wealth in the stock so that no trading is needed inside the region. 

Compared to no-gap market models, the all-in-stock region is new. In the other two regions, the trading strategies are of the same form as the no-gap case, except for that one should use the borrowing rate in the borrowing-money region and the saving rate in the saving-money region. Therefore, our optimal strategy is a mixture of \ctr strategy (as suggested by most continuous-time models) and \dtr strategy (as suggested by models with transaction costs): one does not need to trade in the middle all-in-stock region, and has to continuously trade the stock in the borrowing-money and saving-money regions.

Although both the presence of transaction costs (such as \cite{DN90,DY09,DXZ10}) and the presence of gap between the borrowing and saving rates lead to similar optimal \dtr strategies, the reasons behind are fairly different. In transaction costs models, trading frequently directly increases transaction costs, so one should not trade all the time, leading to the existence of no-trading regions. According our gap-market model, one should not borrow money when the marginal cost of borrowing at a high rate is higher than the marginal benefit of extra leveraging, and one should not save money when the marginal benefit resulted from the low return is not high enough to compensate for giving up a better return-risk trade-off provided by the stock, so the existence of no-trading region is due to the gap between the borrowing and saving rates.

Mathematically speaking, it is very important to notice that the diverging of the borrowing and saving rates forces the wealth dynamics to become piecewise linear, and no longer linear. As a consequence, the stochastic LQ control theory cannot be applied and new stochastic control theory is called for to solve the problem. Indeed, one can apply stochastic LQ control methods to solve Markowitz's problem only when the value function of the problem is of a quadratic form, in which case one reduces to determining the coefficients of quadratic function by solving the so-called Riccati equation. Because of the diverging of the borrowing and saving rates in our model, the associated HJB equation is a fully nonlinear PDE and does not admit any solution in quadratic form. Hence, the problem cannot reduce to solving some Riccati equation. By contrast, because of the infinite time horizon setting, the HJB equation in Fleming and Zariphopoulou \cite{FZ06} is an ordinary differential equation, which is easier to study than our PDE. Because of this, the method of \cite{FZ06} cannot be applied to our model.
Instead, we adopt the PDE argument used in Dai and Yi \cite{DY09} and Guan \cite{G18} to solve our problem. We first transform the associated HJB equation into a semi-linear parabolic PDE through an intuitive argument. Adopting some standard PDE tools including the truncation method, the Leray-Schauder fixed point theorem, the embedding theorem and the Schauder estimation, we derive a solution to the semi-linear parabolic PDE, from which we eventually construct a $C^{3,2}$ smooth solution to the original HJB equation. Different from \cite{FLL10,DY09,G18}, we show that the constructed solution is indeed the value function to our mean-variance problem through a verification theorem. An optimal feedback strategy is also obtained during this process (where the smoothness of the value function plays an important role). The first-order smoothness of the borrowing and saving boundaries are obtained as well under some slightly stronger conditions on the market parameters.

The reminder of this paper is organized as follows. In Section \ref{sec:mf}, we formulate a mean-variance portfolio selection problem under different borrowing and saving rates. In Section \ref{sec:main}, we present our main theoretical results including the smoothness of the value function and provide an optimal control to the problem. Numerical examples are also provided to justify our theoretical results. Sections \ref{sec:Eq}-\ref{sec:V_solu} are devoted to the proofs of the main technical results. We first derive a semi-linear parabolic PDE from the original fully nonlinear HJB equation through an intuitive argument in Section \ref{sec:Eq}; then show that the parabolic PDE has a classical solution by PDE method in Section \ref{sec:hjb}; and Section \ref{sec:V_solu} completes the proof of the main results presented in Section \ref{sec:main}. Some concluding remarks are given in Section \ref{sec:cr}.


\section{Model Formulation}\label{sec:mf}
We call a filtered complete probability space $(\Omega, \cF, \P, \{\cF_t\}_{t\geq0})$ the financial market. And assume that the filtration $\{\cF_t\}_{t\geq0}$ is generated by a standard one-dimensional Brownian motion $\{W_t, t\geq 0\}$ defined in the probability space, argumented with all $\P$-null sets.

The market consists of a risk-free money account and a continuously traded stock. The saving rate and the borrowing rate of the money account are different, denoted by $r_1$ and $r_2$, respectively. Economically speaking, the borrowing rate shall be higher than the saving rate. The stock price process $S^1>0$ follows a geometric Brownian motion:
\begin{equation*} 
{\rm d}S_t^1=S_t^1\big(\mu{\rm d}t+\sigma {\rm d}W_t\big),\end{equation*}
where $\mu$ is the appreciation rate, and $\sigma$ is the volatility rate of the stock. We assume that the market parameters $r_1$, $r_2$, $\mu$ and $\sigma$ are constants and satisfy $\sigma>0$ and
\begin{equation}\label{murr}
\mu >r_2>r_1.
\end{equation}
Remind that $r_1$, $r_2$, $\mu$ are not necessary to be positive, which happens in many financial markets right now.

Consider an agent (``She'') faced with an initial endowment $x$ and an investment horizon $[t,T]$. Let $X_s$ and $\pi_s$ denote her total wealth and dollar amount invested in the stock at time $s$, respectively. When $X_s>\pi_s$, the agent saves the extra money of the amount $X_s-\pi_s$ in the money account to earn interests at the saving rate $r_1$; whereas when $X_s<\pi_s$, the agent borrows the money of the amount $\pi_s-X_s$ from the money account at the borrowing rate $r_2$. Assume that the trading of shares is self-financed and takes place continuously, and there are no transaction costs or taxes. Then the wealth process $X_s$ of the agent satisfies the following stochastic differential equation (SDE):
\begin{equation}\label{X_eq}
\left\{\begin{array}{rl}
{\rm d}X_s&=\big[\big(r_1\chi_{X_s>\pi_s}+r_2\chi_{X_s<\pi_s}\big)(X_s-\pi_s)+\mu\pi_s\big]{\rm d}s+\sigma\pi_s{\rm d}W_s,\;\; t\leq s\leq T, \\ [2mm]
X_t&=x.
\end{array}\right.\end{equation}
Here $\chi_{S}$ is the indicator function for a statement $S$: it is equal to 1 if the statement $S$ is true, and 0 otherwise.

We call the process $\pi=\{\pi_s\}_{s\in[t,T]}$, a portfolio of the agent. Define the set of admissible portfolios as
\begin{align*}
\Pi_t:=L^2_{\cal F}([t,T];\R),
\end{align*}
where $L^2_\cF([t,T];\mathbb{R})$ denotes the set of all $\mathbb{R}$-valued, $\cF_s$-progressively measurable stochastic processes $f(\cdot)$ on $[t,T]$ with $\E\int_t^T|f(s)|^2\ds<+\infty$. For any admissible portfolio $\pi\in\Pi_t$, the SDE \eqref{X_eq} admits a unique strong solution $X_{\cdot}$ on $[t,T]$.

Given a constant target $d>0$, the agent's objective is to find an admissible portfolio $\pi^*\in\Pi_t$ to solve the following portfolio selection problem
\begin{align}\label{value}
V(x,t)=\inf\limits_{\pi\in \Pi_t}\E \big[(X_T-d)^2 \mid X_t=x \big],\quad (x,t)\in \TQ,
\end{align}
where
$$\TQ=\big\{(x,t)\mid \;xe^{r_1(T-t)} <d,\;0\leq t< T\big\}.$$
If such an admissible portfolio $\pi^*\in \Pi_t$ exists, we call it an optimal portfolio for the problem \eqref{value}.
The agent's target $d$ shall be higher than the outcome of saving all her money in the money account, so we put the constraint $xe^{r_1(T-t)}<d$, leading to the above admissible region $\TQ$.

The main aim of this paper is to determine the optimal value function $V(x,t)$ and find an optimal portfolio to the stochastic control problem \eqref{value}.

\begin{remark}
The standard Markowitz's problem can be formulated as
\begin{align}\label{stvalue}
\inf\limits_{\pi\in \Pi_t}&\quad \mathrm{Var}(X_T),\quad
\mathrm{s.t.}\quad \E\big[X_T\big]=z,\quad X_t=x,
\end{align}
where $z >0$ and $xe^{r_1(T-t)}<z$. Let $V_{MV}(x,t, z)$ denote its optimal value.
Then the set
\[\Big\{\big(\sqrt{V_{MV}(x,t, z)},\; z\big): z>xe^{r_1(T-t)}\Big\}\]
is called the efficient mean-variance frontier. By the Lagrange duality theorem (see Luenberger \cite{L69}), we have
\begin{align}\label{Lagrangeduality}
V_{MV}(x,t, z)=\sup_{d>z} \big[V(x,t,d)-(d-z)^2\big],
\end{align}
where $V(x,t,d)=V(x,t)$ defined by \eqref{value}. Indeed, the optimal $d$ is determined by
\begin{align}\label{Lagrangeduality2}
\frac{\partial }{\partial d}V(x,t,d)=2(d-z),\quad d>z.
\end{align}
In order to determine the efficient mean-variance frontier, it is unnecessary to solve the optimization problem in \eqref{Lagrangeduality}, or equivalently, to solve \eqref{Lagrangeduality2}. Indeed, by \eqref{Lagrangeduality} and \eqref{Lagrangeduality2}, the efficient mean-variance frontier can be expressed as
\[\bigg\{\bigg(\sqrt{V(x,t, d)-\frac{1}{4}\Big(\frac{\partial }{\partial d}V(x,t,d)\Big)^2}, \; d-\frac{1}{2}\frac{\partial }{\partial d}V(x,t,d)\bigg): d>xe^{r_1(T-t)}\bigg\}.\]
By the above relationship, it suffices to solve the portfolio selection problem \eqref{value} in order to solve the standard Markowitz's problem \eqref{stvalue}. Our proceeding analysis will also show that numerically solving $V(x,t,d)$ for each fixed $d$ can reduce to solving the approximation equation in a bounded domain \eqref{wNN_pb}. Clearly, the latter can be computed by standard such as finite difference method.
\end{remark}

\begin{remark}
When bankruptcy is prohibited in the market, we need to replace $\TQ$ by a bounded domain
$$\big\{(x,t)\mid \;0<xe^{r_1(T-t)} <d,\;0\leq t< T\big\}.$$
Meanwhile, we need to put an extra boundary condition $V(0,t)=d^2$, $0\leq t< T$ into the HJB equation \eqref{V_pb} below. Our argument, after minor adjustment, still works for that case. We encourage the interested reader to write down the details.
\end{remark}

\section{Main Results}\label{sec:main}
Using the standard viscosity theory (see, e.g. Grandall and Lions \cite{CL83}, Yong and Zhou \cite{YZ99}), one can prove that the value function of \eqref{value} is a viscosity solution to the following HJB equation with boundary and terminal conditions:
\begin{align}\label{V_pb}
\left\{
\begin{array}{ll}
V_t+\inf\limits_{\pi}\Big(\frac{1}{2}\si^2\pi^2V_{xx}+\Big((r_1\chi_{\pi<x}+r_2\chi_{\pi> x})(x-\pi)+\mu\pi\Big)V_x\Big)=0\quad \hbox{in}\quad \TQ,\\[5mm]
V(e^{-r_1(T-t)} d,t)=0,\quad 0\leq t<T,\\[5mm]
V(x,T)=(x-d)^2,\quad x< e^{-r_1(T-t)} d.
\end{array}
\right.\end{align}
This paper does not adopt the viscosity approach because viscosity solution usually does not lead to good smoothness of the value function. Instead, we will prove that the above HJB equation \eqref{V_pb} admits a classical solution $V$ (see the precise definition in Theorem \ref{theo:V} below) by PDE method directly. This together with a verification result (see \thmref{veri} below) can guarantee that $V$ is the value function of the problem \eqref{value}.

\begin{theorem}[Solvability of the HJB Equation \eqref{V_pb}]\label{theo:V}
There exists a solution 
$$V\in C^{3,2}\big(\overline{\TQ}\setminus\{x=e^{-r_1 (T-t)}d\}\big)\bigcap C\big(\overline{\TQ}\big)$$ 
to the HJB equation \eqref{V_pb} such that
\begin{align}\label{Vx_b}
&V_x<0 ,\\ \label{Vxx_b}
&V_{xx}>0\end{align}
in $\TQ$, and
\begin{align}\label{V_lim}
\lim\limits_{x\rightarrow e^{-r_1 (T-t)}d-}V_x=0,\quad \lim\limits_{x\rightarrow-\infty}V_x=-\infty,\quad \forall\;t\in[0,T].\end{align}
\end{theorem}
\begin{proof}
We leave the proof to Section \ref{proofv}.
\end{proof}
Figure \ref{fig1} illustrates the function $V(x,t)$, based on \thmref{theo:V}.
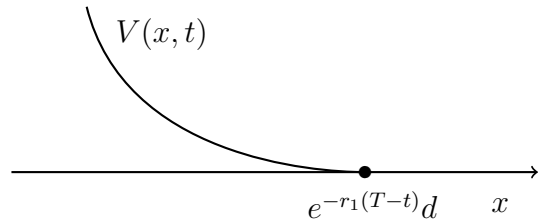
\begin{figure}[H]
\begin{center}
\begin{tikzpicture}[thick]
\draw[thick] (2, 2.2) .. controls (2.5, 0.2) and (5, 0) .. (5.7,0);
\draw[->, black] (1,0)--(8,0);
\node[below] at (5.8,-0.1) {$e^{-r_1 (T-t)}d$};
\node[below] at (5.7,0.23) {$\bullet$};
\node[below] at (3, 2.2) {$V(x,t)$};
\node[below] at (7.5,-0.2) {$x$};
\end{tikzpicture}
\end{center}
\caption{The solution of the HJB Equation \eqref{V_pb}: $V(x,t)$ with fixed $t$.}
\label{fig1}
\end{figure}

\subsection{Optimal Portfolio}
Let $V$ be given in \thmref{theo:V} and we divide the whole trading area 
$$\TQ=\big\{(x,t)\mid \;xe^{r_1(T-t)} <d,\;0\leq t< T\big\}$$ 
into three regions:
\begin{align*}
\hbox{\bf Borrowing-money Region $\fB$} &:=\Big\{(x,t)\in \TQ\;\Big|\;-\frac{\mu-r_2}{\si^2}\frac{V_{x}}{V_{xx}}>x\Big\}, \\[2mm]
\hbox{\bf All-in-stock Region $\fM$} &:=\Big\{(x,t)\in \TQ\;\Big|\;-\frac{\mu-r_2}{\si^2}\frac{V_{x}}{V_{xx}}\leq x\leq-\frac{\mu-r_1}{\si^2}\frac{V_{x}}{V_{xx}}\Big\}, \\[2mm]
\hbox{\bf Saving-money Region $\fS$}&:=\Big\{(x,t)\in \TQ\;\Big|\;-\frac{\mu-r_1}{\si^2}\frac{V_{x}}{V_{xx}}<x\Big\}.
\end{align*}
The following result shows that they are connected regions and separated by two free boundaries.
\begin{proposition}[Optimal Trading Regions]\label{theo:free_boundary}
We have
\begin{align*}
\fB&=\Big\{(x,t)\mid x<B(t),\;t\in[0,T)\Big\},\\
\fM&=\Big\{(x,t)\mid B(t)\leq x\leq L(t),\;t\in[0,T)\Big\},\\
\fS&=\Big\{(x,t)\mid L(t)<x<e^{-r_1(T-t)} d,\;t\in[0,T)\Big\},
\end{align*}
where $B(\cdot)$ and $L(\cdot)$ are respectively called the borrowing and saving boundaries, defined by
\begin{align*}
B(t):=V_x^{-1}(\cdot,t)(-e^{b(T-t)}),\quad L(t):=V_x^{-1}(\cdot,t)(-e^{l(T-t)}),
\end{align*}
with $V_x^{-1}(\cdot,t)$ being the inverse with respect to (w.r.t.) the spatial argument $x$, and the two functions $b(\cdot)$ and $l(\cdot)$ are given by \eqref{boundaryb} and \eqref{boundaryl}. Moreover, we have the estimate
\[0<B(t)<L(t)<e^{-r_1(T-t)} d,\quad t\in[0,T],\]
and the terminal values
\[B(T)=\frac{\mu-r_2}{\si^2+\mu-r_2}d,\quad L(T)=\frac{\mu-r_1}{\si^2+\mu-r_1}d.\] 
\end{proposition}

\begin{proof}
We leave the proof to Section \ref{sec:freeboundary}.
\end{proof}

Figure \ref{regions} illustrates the borrowing and saving boundaries, based on \propref{theo:free_boundary}.

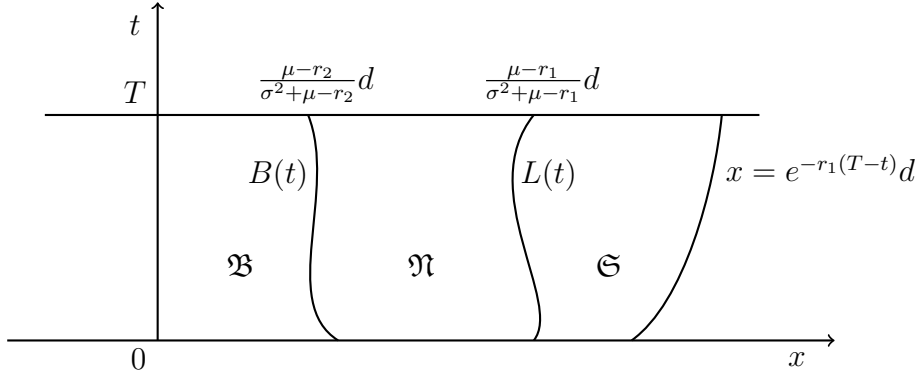
\begin{figure}[H]
\begin{center}
\begin{tikzpicture}[thick]
\draw[->, black] (-2,0)--(9,0);
\draw[->, black] (0,0)--(0,4.5);
\draw[-, black] (-1.5,3)--(8,3);
\node at (-0.3, 3.3) {$T$};
\node at (-0.3, 4.2) {$t$};
\draw[thick] (5, 3) .. controls (4.2, 2) and (5.4, 0.5) .. (5,0);
\draw[thick] (2, 3) .. controls (2.4, 2) and (1.6, 0.5) .. (2.4,0);
\node at (1.1, 1) {$\fB$};
\node at (3.5, 1) {$\fM$};
\node at (6, 1) {$\fS$};
\node at (1.6, 2.2) {$B(t)$};
\node at (5.2, 2.2) {$L(t)$};

\node[above] at (2.1, 3){$\frac{\mu-r_2}{\si^2+\mu-r_2}d$};
\node[above] at (5.1, 3){$\frac{\mu-r_1}{\si^2+\mu-r_1}d$};

\node[below] at (8.5, 0) {$x$};
\node[left] at (0,-0.24) {$0$};

\draw[thick] (6.3, 0) .. controls (7, 0.5) and (7.4, 2) .. (7.5, 3);

\node[right] at (7.4, 2.3) {$x=e^{-r_1(T-t)} d$};
\end{tikzpicture}
\end{center}
\vspace{0pt}
\caption{The optimal trading regions.} \label{regions}
\end{figure}

\begin{remark}
In fact, 
\[B(t)=\sup\Big\{x\;\Big|\;-\frac{\mu-r_2}{\si^2}\frac{V_{x}}{V_{xx}}>x,\; (x,t)\in \TQ\Big\},\]
and
\[L(t)=\inf\Big\{x\;\Big|\;-\frac{\mu-r_1}{\si^2}\frac{V_{x}}{V_{xx}}<x, \; (x,t)\in \TQ\Big\}.\]
\end{remark}

We will also establish the first-order smoothness of the boundaries $B(\cdot)$ and $L(\cdot)$ under certain conditions (see \propref{prop:C1}).

\begin{theorem}[Verification Theorem]\label{veri}
The function $V$ given in \thmref{theo:V} is the same as the value function $V$ defined by \eqref{value}.
Moreover, the optimal portfolio to the problem \eqref{value}, given in the feedback form, is
\[
\pi(x,t)=
\begin{cases}
-\frac{\mu-r_2}{\si^2}\frac{V_{x}(x,t)}{V_{xx}(x,t)}, & (x,t)\in \fB, \\[3mm]
x, & (x,t)\in \fM, \\[3mm]
-\frac{\mu-r_1}{\si^2}\frac{V_{x}(x,t)}{V_{xx}(x,t)}, & (x,t)\in \fS.
\end{cases}
\]
\end{theorem}
\begin{proof}
We leave the proof to Section \ref{sec:veri}.
\end{proof}

We have the following financial findings from the above theoretical results. When one's wealth is far from her target (i.e. $x<B(t)$), she must borrow money to invest in the stock so as to maximize the chance to achieve her goal $d$. By contrast, if her wealth is sufficiently close to her target (i.e. $x>L(t)$), she does not need to invest all her wealth in the stock and can save some in the money account to reduce her risk. In the middle range (i.e. $B(t)\leq x\leq L(t)$), she does not need to borrow or save money, and shall invest all her wealth in the stock so that no trading happens there. Therefore, we see the optimal strategy is a mixture of \ctr strategy in the first two scenarios (as suggested by most continuous-time models) and \dtr strategy in the last scenario (as suggested by models with transaction costs).

Also, the optimal portfolio is long in the stock in all scenarios, so it is never optimal to short sale the stock. As a consequence, the portfolio is still optimal if we restrict us to the control set with no-shorting constraint:
\begin{align*}
\Big\{\pi_s\in L^2_{\cal F}([t,T];\R)\;\Big|\; \pi_s\geq 0,\; s\in[t,T]\Big\}
\end{align*}
in the problem \eqref{value}.

When $r_2\to r_1$, the optimal feedback portfolio reduces to
\[\pi(x,t)=-\frac{\mu-r_1}{\si^2} \frac{V_{x}(x,t)}{V_{xx}(x,t)}.\]
This recovers the classical optimal portfolio when there is no gap between the borrowing and saving rates. In this case, continuously tradings happen all the time, and the all-in-stock region is a zero measure set.

\subsection{Numerical Study}

We now present numerical examples to depict the behaviors of the borrowing and saving boundaries as well as the optimal value when the borrowing and saving rates change.

\begin{figure}[H]
\begin{center}
\subfigure[Free boundaries $B(\cdot)$ and $L(\cdot)$ when $r_1=0.02,\;r_2=0.08$]
{\includegraphics[width=0.47\linewidth]{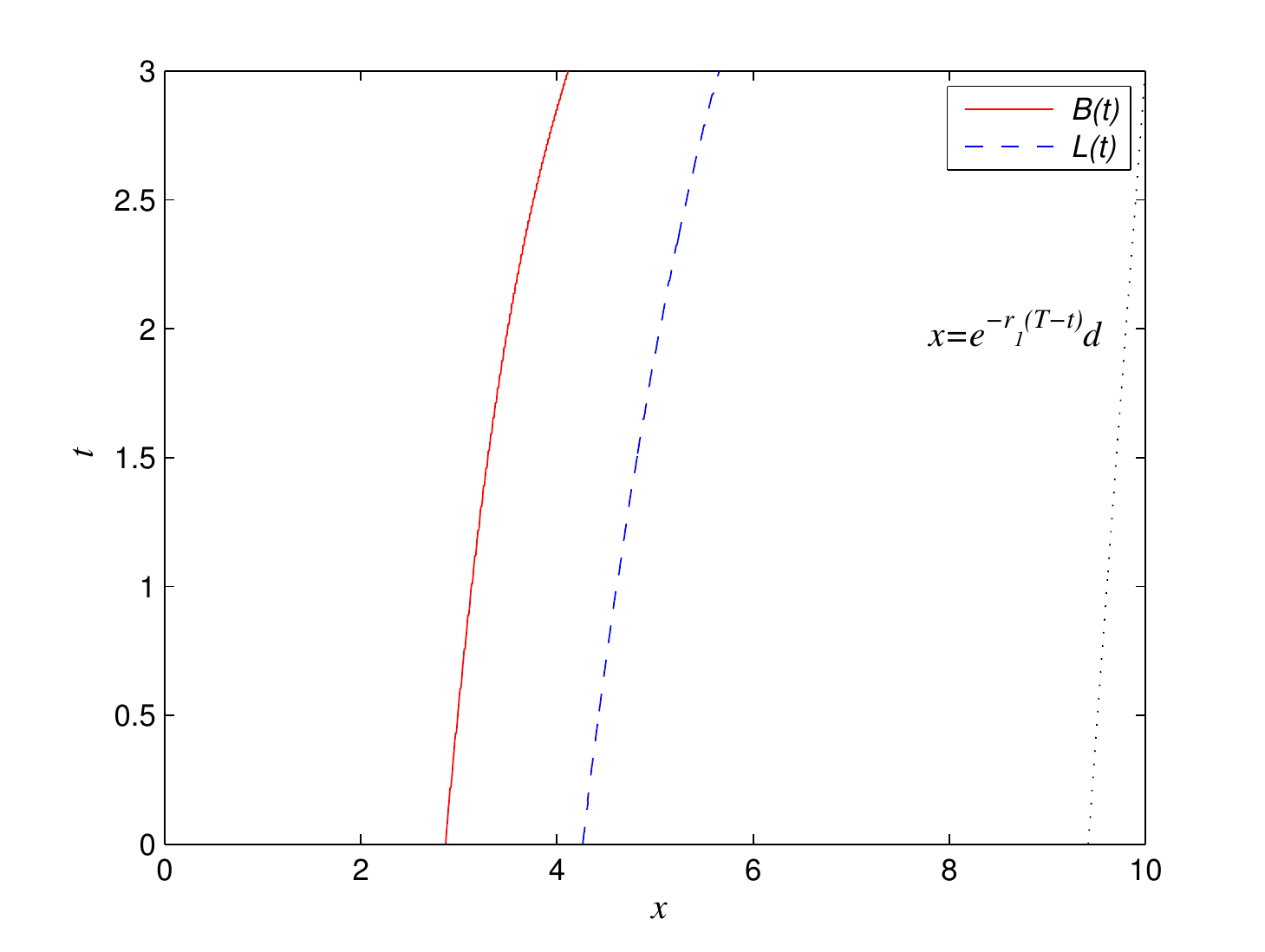} }
\subfigure[Free boundaries $B(\cdot)$ and $L(\cdot)$ when $r_1=0.03,\;r_2=0.07$]
{\includegraphics[width=0.47\linewidth]{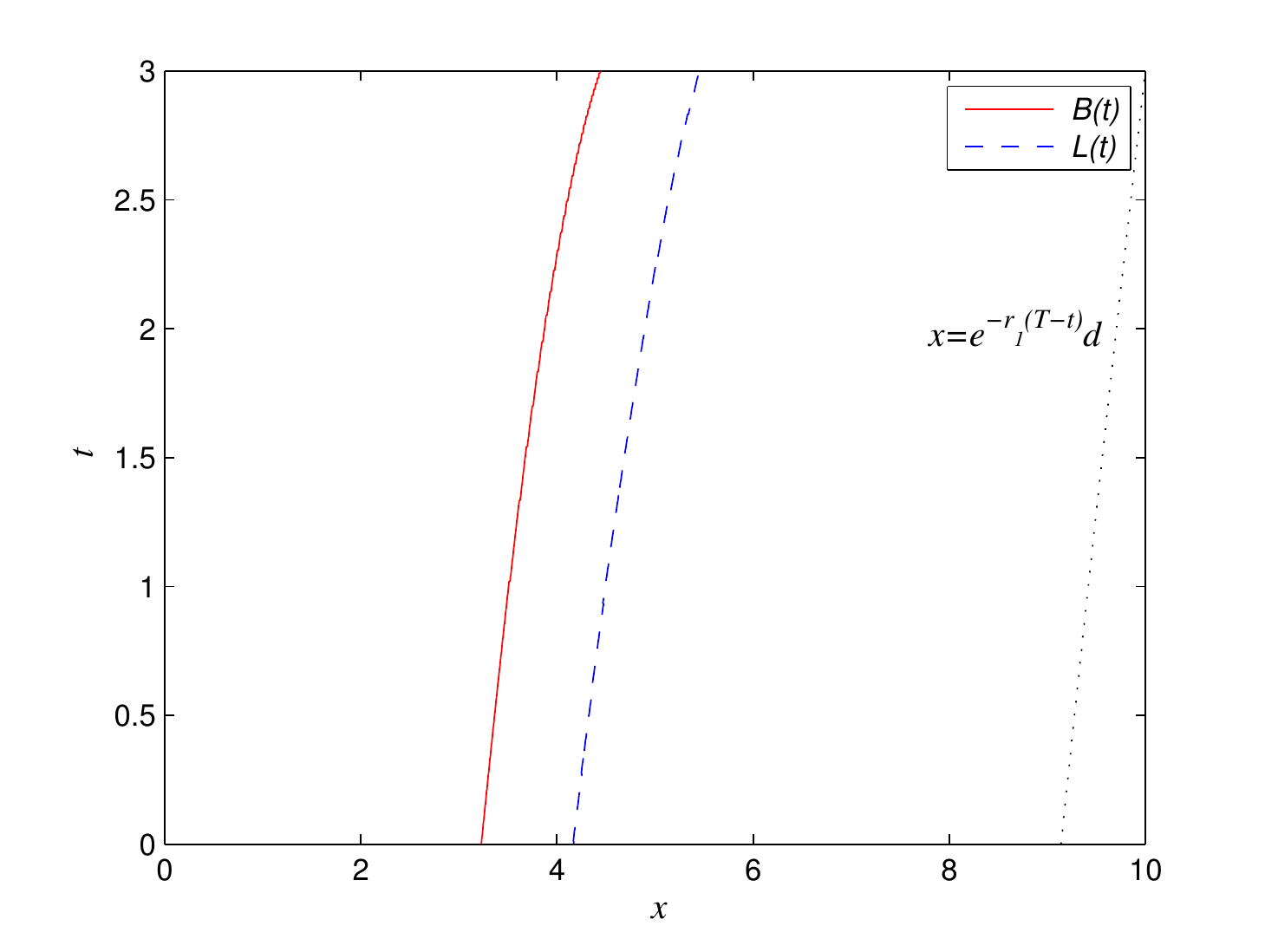} }
\end{center} 
\begin{center}
\subfigure[Free boundaries $B(\cdot)$ and $L(\cdot)$ when $r_1=0.04,\;r_2=0.06$]
{\includegraphics[width=0.47\linewidth]{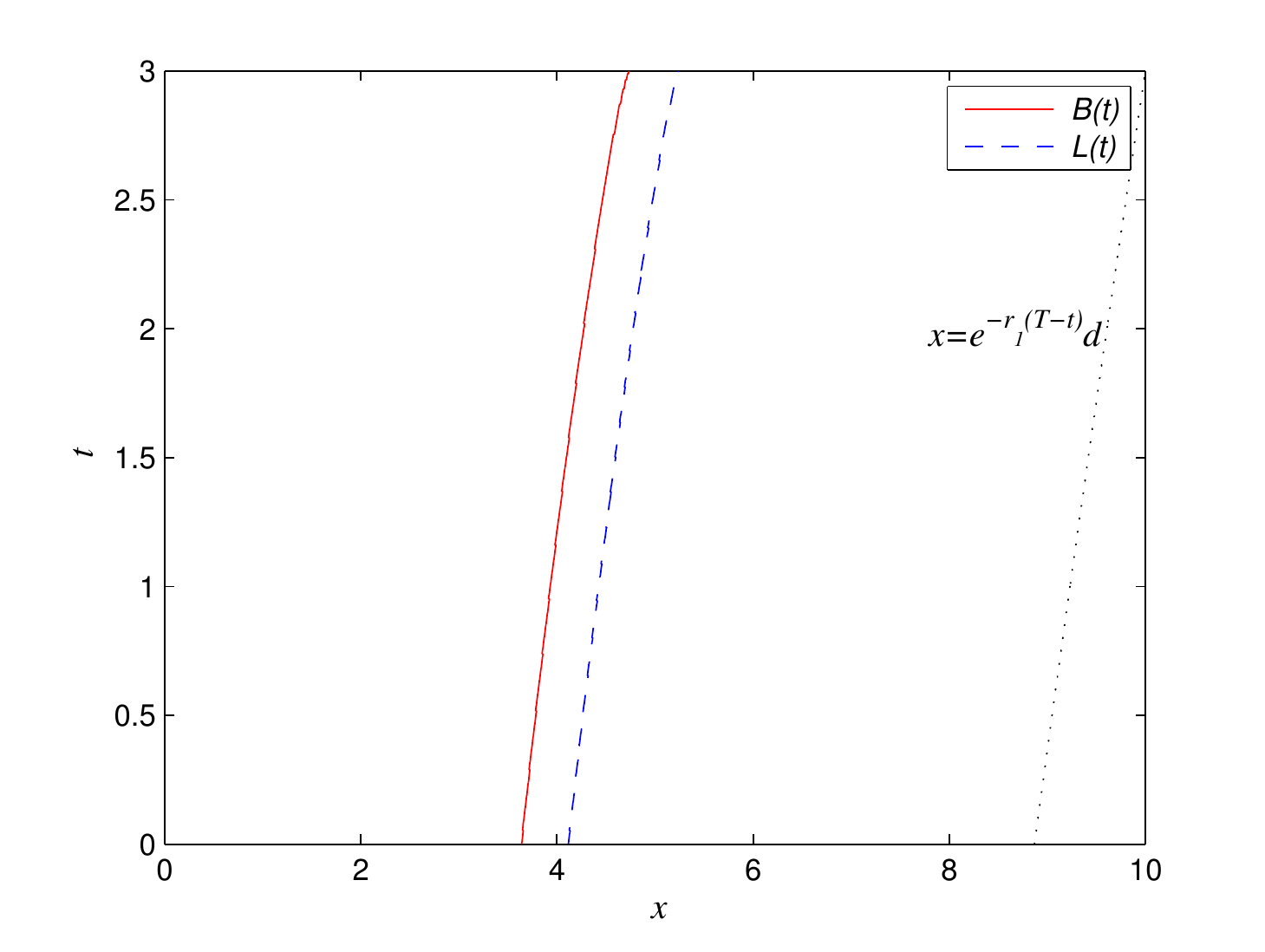} }
\subfigure[Free boundaries $B(\cdot)$ and $L(\cdot)$ when $r_1=r_2=0.05$]
{\includegraphics[width=0.47\linewidth]{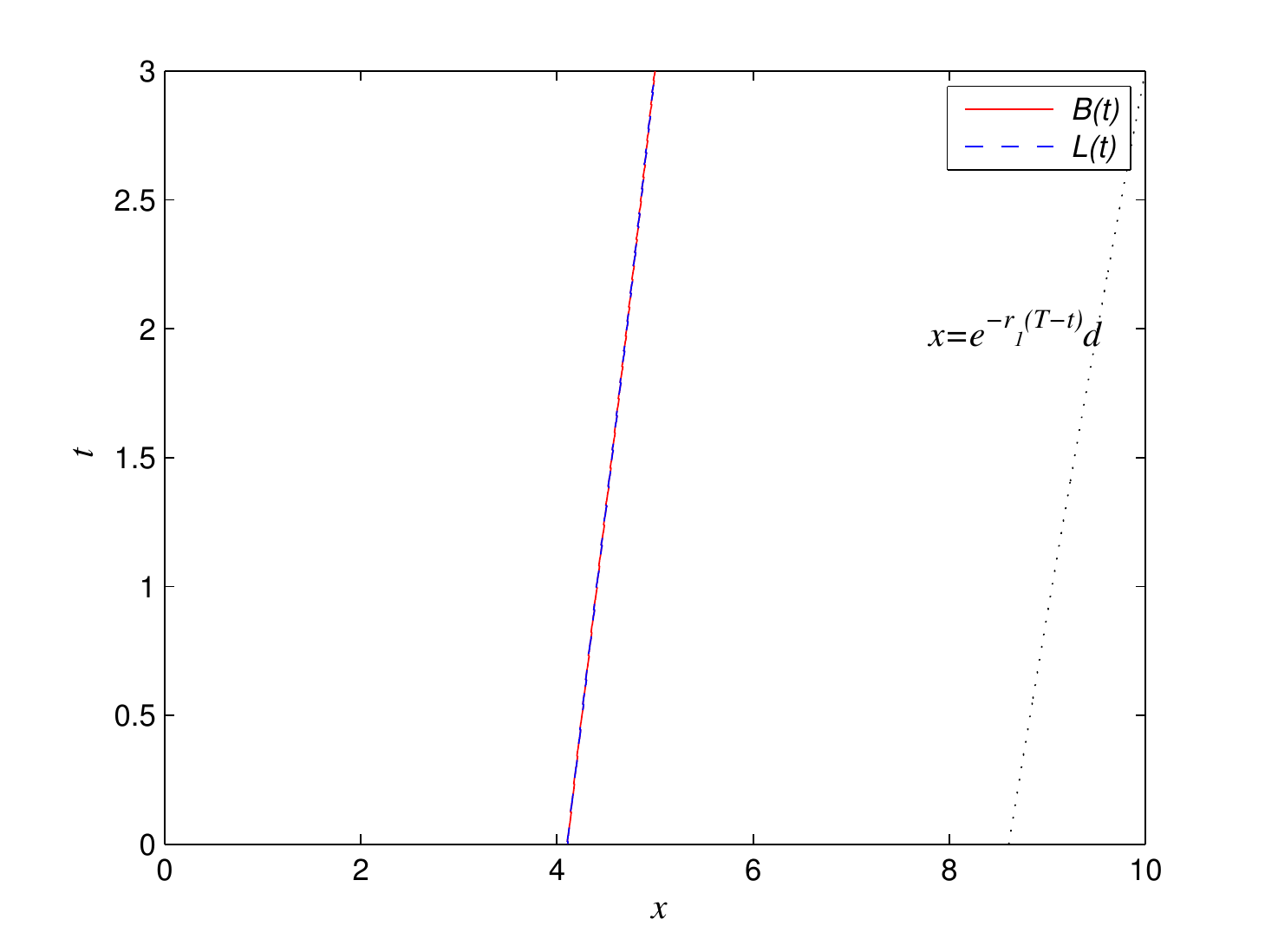} }
\end{center}
\caption{The borrowing and saving boundaries $B(\cdot)$ and $L(\cdot)$ under various saving rates $r_1$ and borrowing rates $r_2$.}\label{fig_FB}
\end{figure}
Figure \ref{fig_FB} plots the saving and borrowing boundaries under various pairs of saving and borrowing rates $(r_1, r_2)$, where $d=10$, $\mu=0.15$, $\si^2=0.10$ and $T=3$ are fixed. The borrowing boundary is always on the left of the saving boundary. As the gap between the borrowing and saving rates becomes smaller, the borrowing boundary moves to the right whereas the saving boundary moves to the left, becoming closer. Thus, Borrowing-money Region $\fB$ and Saving-money Region $\fS$ expand, whereas All-in-stock Region $\fM$ shrinks. The latter becomes a curve when the gap between the borrowing and saving rates disappears, in which case, the optimal trading strategy becomes a continuously trading one all the time. This is consistent with the common financial intuition that one should trade more frequently if the gap becomes smaller, because the marginal benefit of extra leveraging becomes higher as the costs of borrowing money become less.

\begin{figure}[H]
\begin{center}
\subfigure[Value function, $V(x,0)$, against $x$]
{\includegraphics[width=0.47\linewidth]{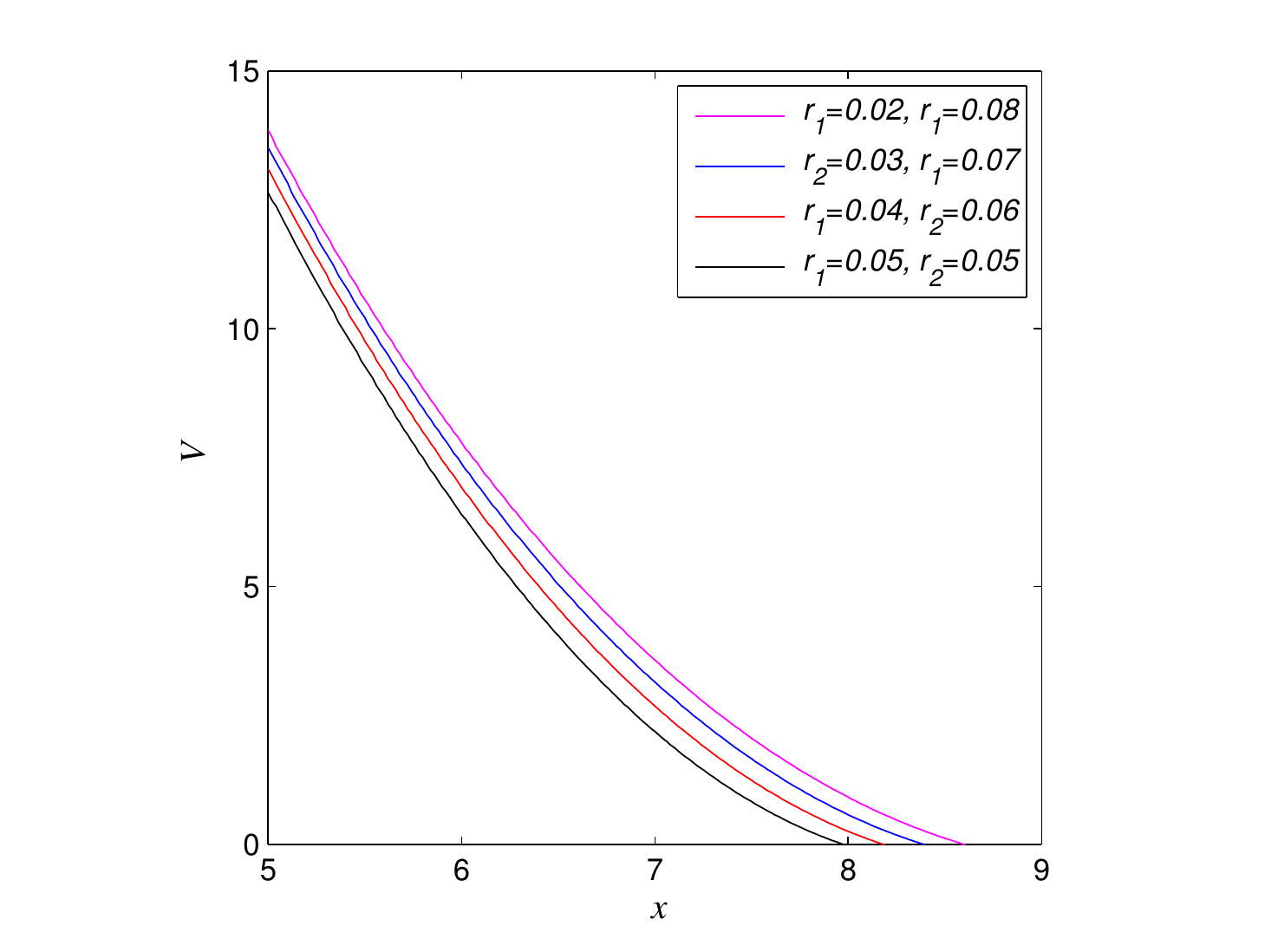} }
\subfigure[Value function, $V(1,t)$, against $t$]
{\includegraphics[width=0.47\linewidth]{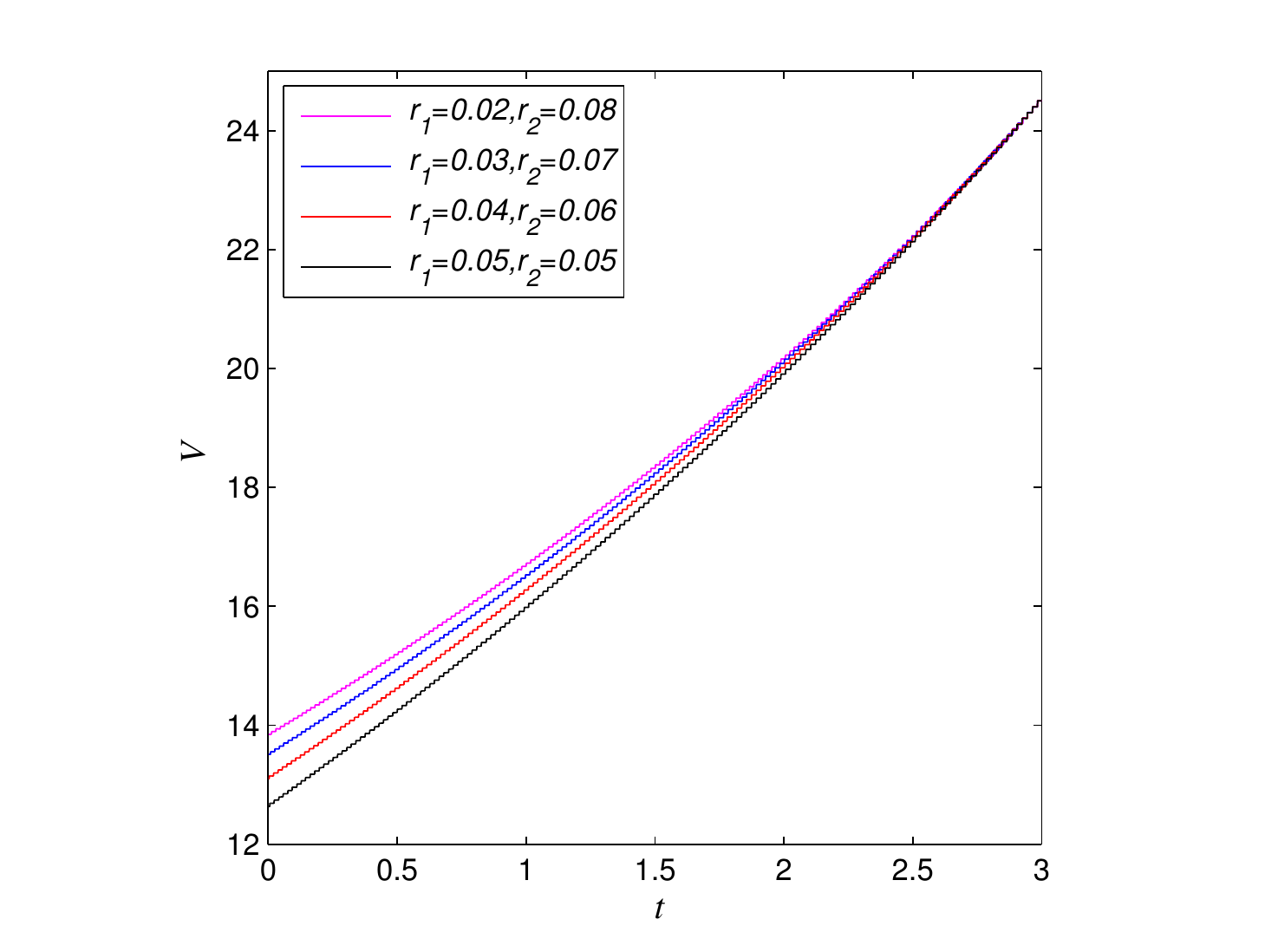} }
\caption{The value function $V$ under various saving rates $r_1$ and borrowing rates $r_2$.}\label{fig_VF}
\end{center}
\end{figure}
The left and right panels of Figure \ref{fig_VF} display, respectively, the value function against the wealth position and time under various pairs of saving and borrowing rates $(r_1, r_2)$, where $d=10$, $\mu=0.15$, $\si^2=0.10$ and $T=3$ are the same as in Figure \ref{fig_FB}. Intuitively speaking, the smaller the gap between the borrowing and saving rates, the smaller the risk (i.e. the value function). Both the left and right panels of Figure \ref{fig_VF} confirm this. We also see from the left panel that the bigger the wealth position the smaller the risk, since bigger wealth position is closer to the fixed target. Meanwhile, the right panel demonstrates that the shorter the time to the maturity, the higher the risk, since shorter time to the maturity means less trading opportunities and less likely to achieve the goal.

\begin{figure}[H]
\begin{center}
\subfigure[Free boundaries $B(\cdot)$ and $L(\cdot)$ when $\mu=0.20$]
{\includegraphics[width=0.47\linewidth]{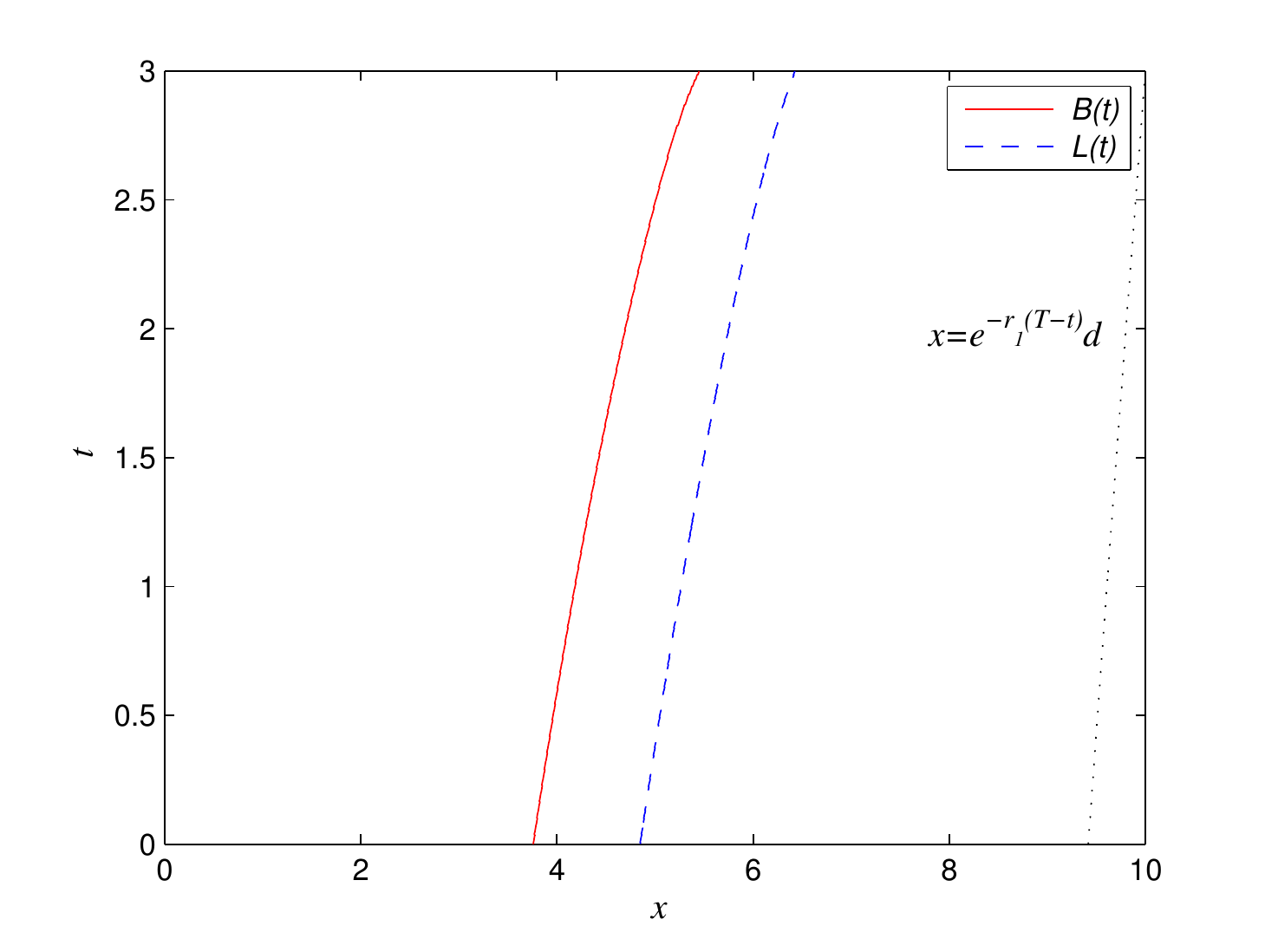} }
\subfigure[Free boundaries $B(\cdot)$ and $L(\cdot)$ when $\mu=0.25$]
{\includegraphics[width=0.47\linewidth]{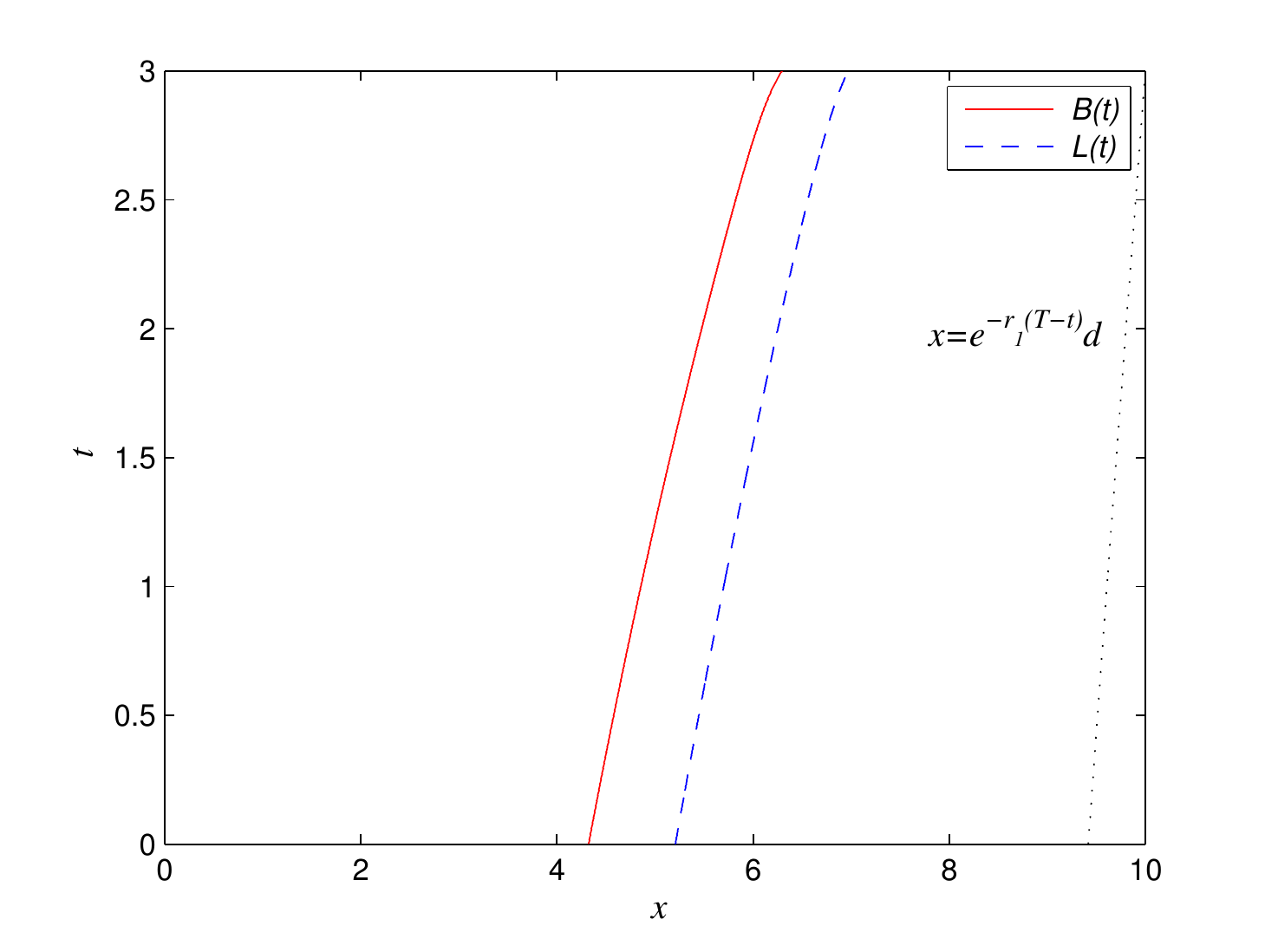} }
\subfigure[Free boundaries $B(\cdot)$ and $L(\cdot)$ when $\mu=0.30$]
{\includegraphics[width=0.47\linewidth]{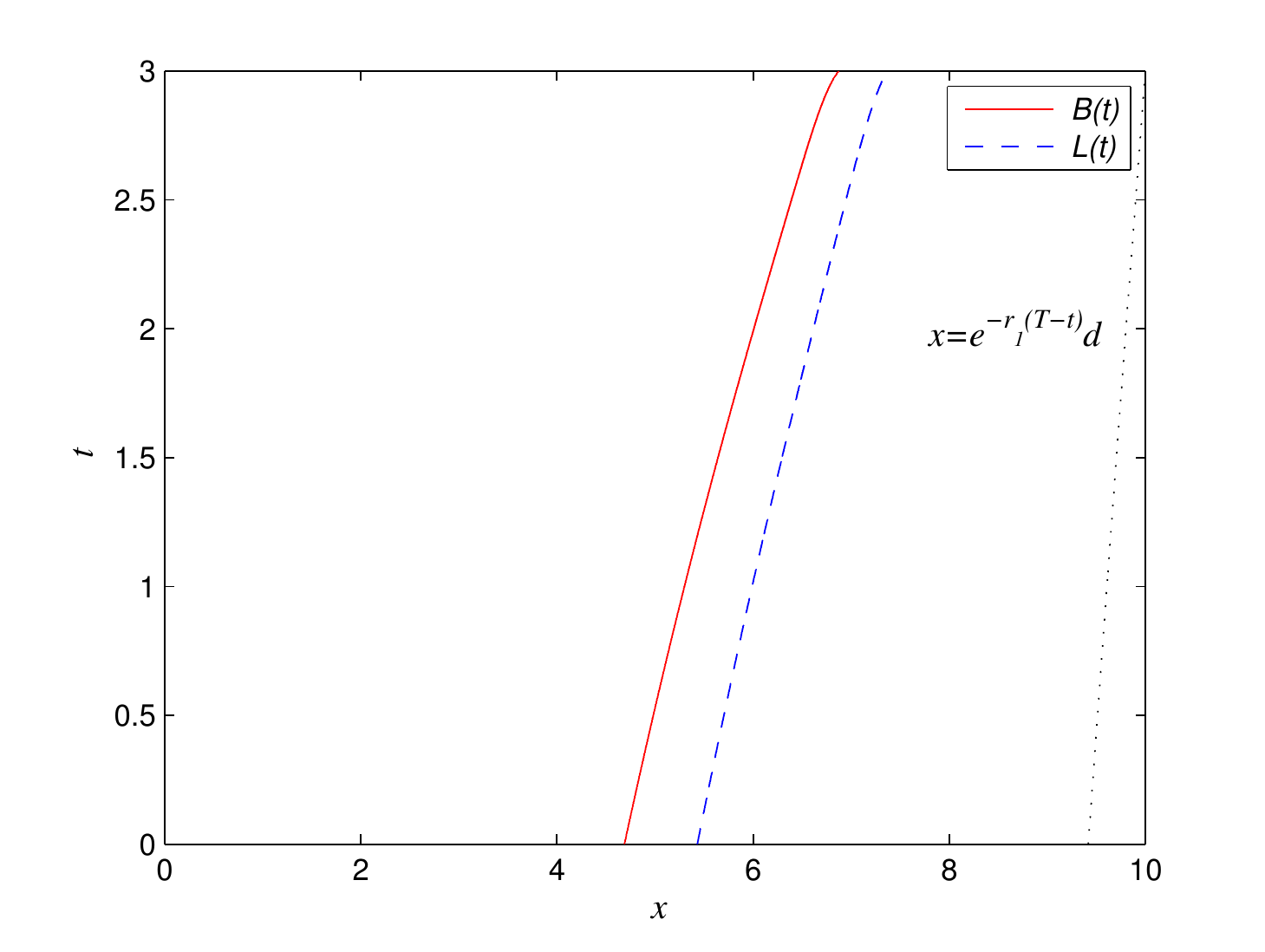} }
\subfigure[Free boundaries $B(\cdot)$ and $L(\cdot)$ when $\mu=0.35$]
{\includegraphics[width=0.47\linewidth]{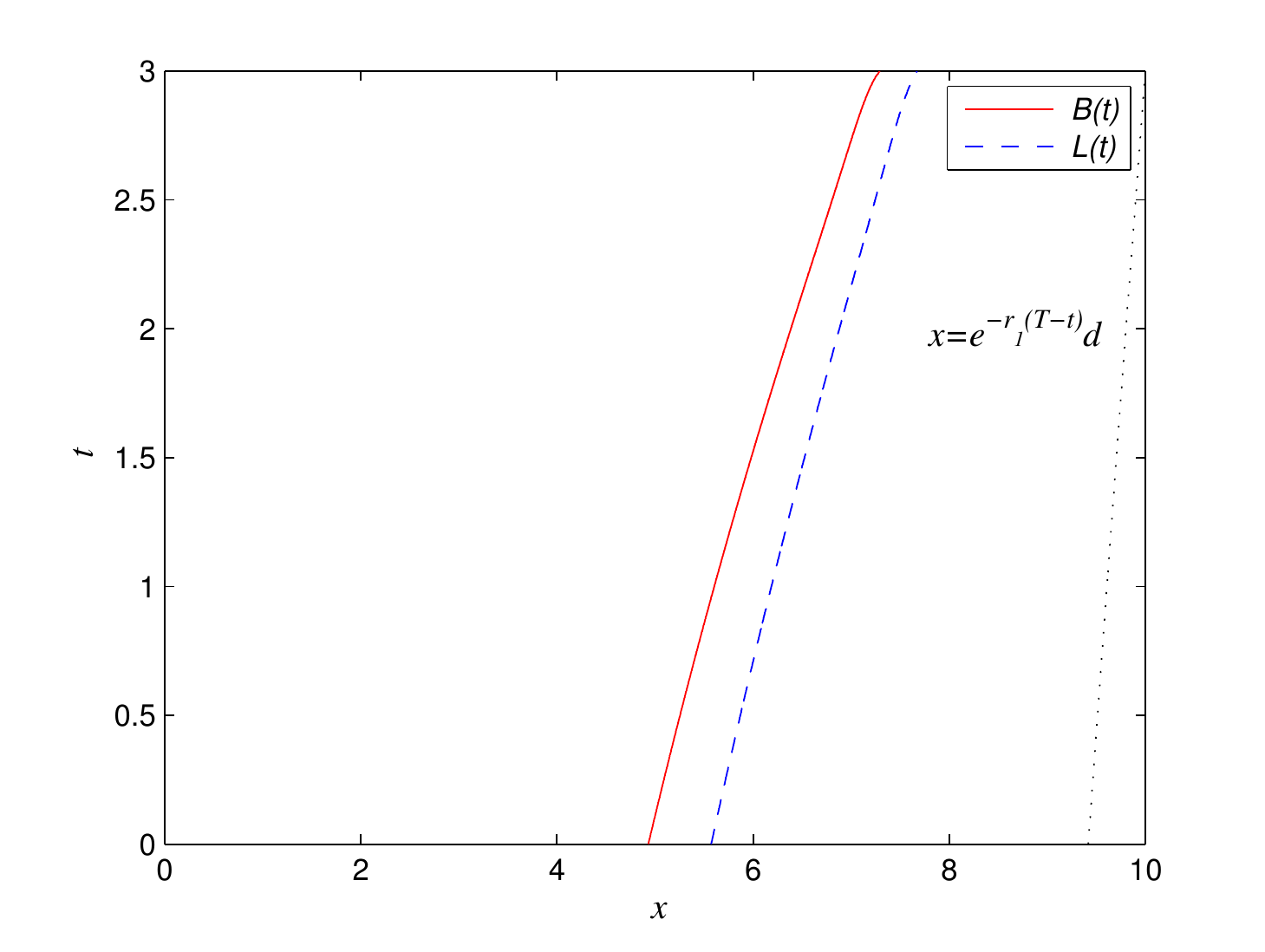} }
\end{center}
\caption{The borrowing and saving boundaries $B(\cdot)$ and $L(\cdot)$ under various return rates $\mu$. }
\label{fig_FB2}
\end{figure}
Figure \ref{fig_FB2} depicts the borrowing and saving boundaries under various return rates $\mu$, where $r_1=0.02$, $r_2=0.08$, $\si^2=0.10$ and $d=10$ are fixed. These pictures show that the two boundaries move to the right as $\mu$ increases. It seems that the closer the time to the maturity day, the faster the speed to move to the right. As the return rate of the stock increases, the role of the gap should become less important, encouraging one in a bad wealth position to borrow money to invest in the stock, discouraging one in a good wealth position to save money. Hence, as $\mu$ increases, Borrowing-money Region $\fB$ shall expand, whereas Saving-money Region $\fS$ and All-in-stock Region $\fM$ shall shrink. The above pictures confirm this financial intuition. 
\begin{figure}[H]
\begin{center}
\subfigure[Free boundaries $B(\cdot)$ and $L(\cdot)$ when $\si^2=0.30$]
{\includegraphics[width=0.47\linewidth]{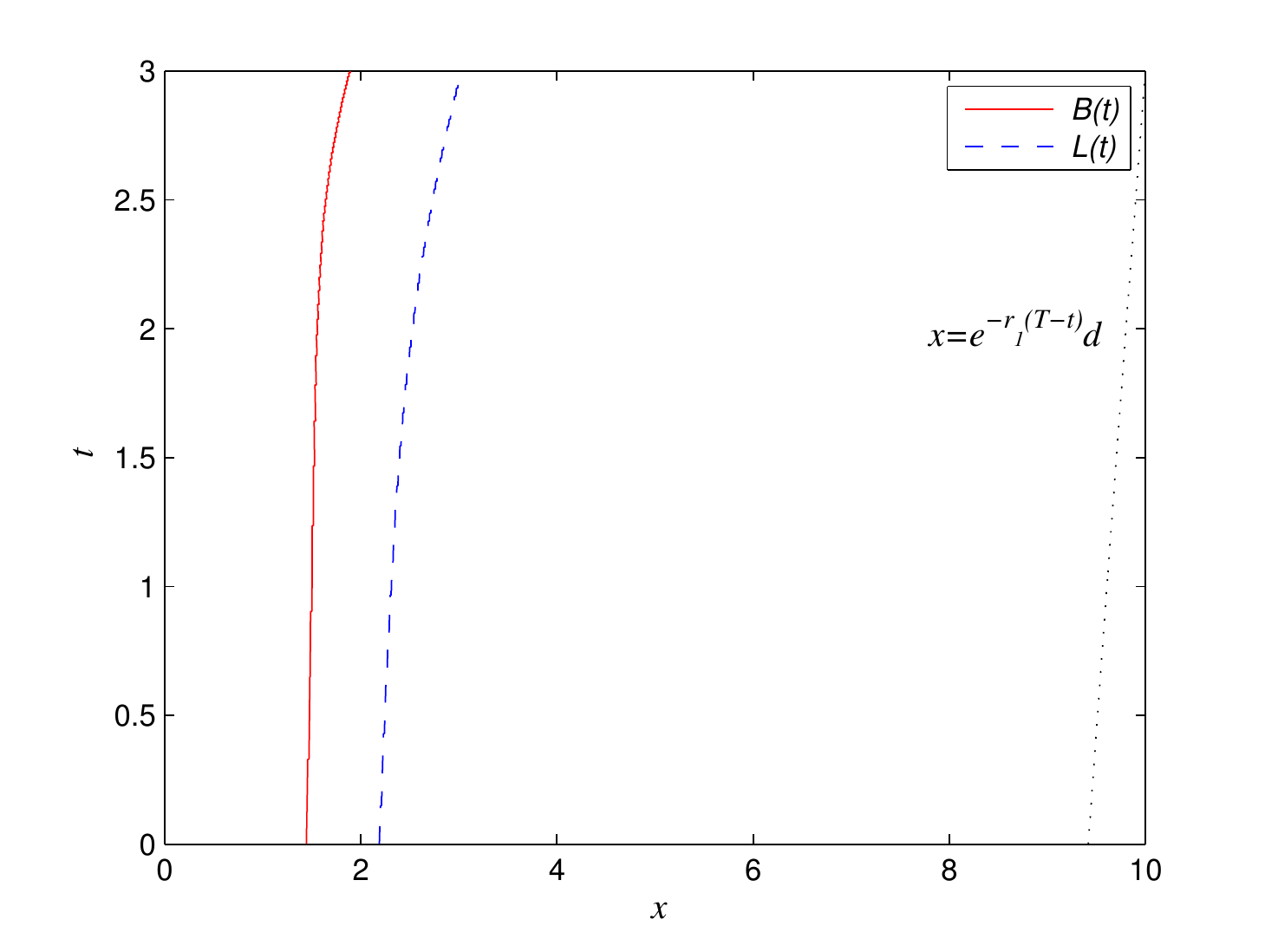} }
\subfigure[Free boundaries $B(\cdot)$ and $L(\cdot)$ when $\si^2=0.25$]
{\includegraphics[width=0.47\linewidth]{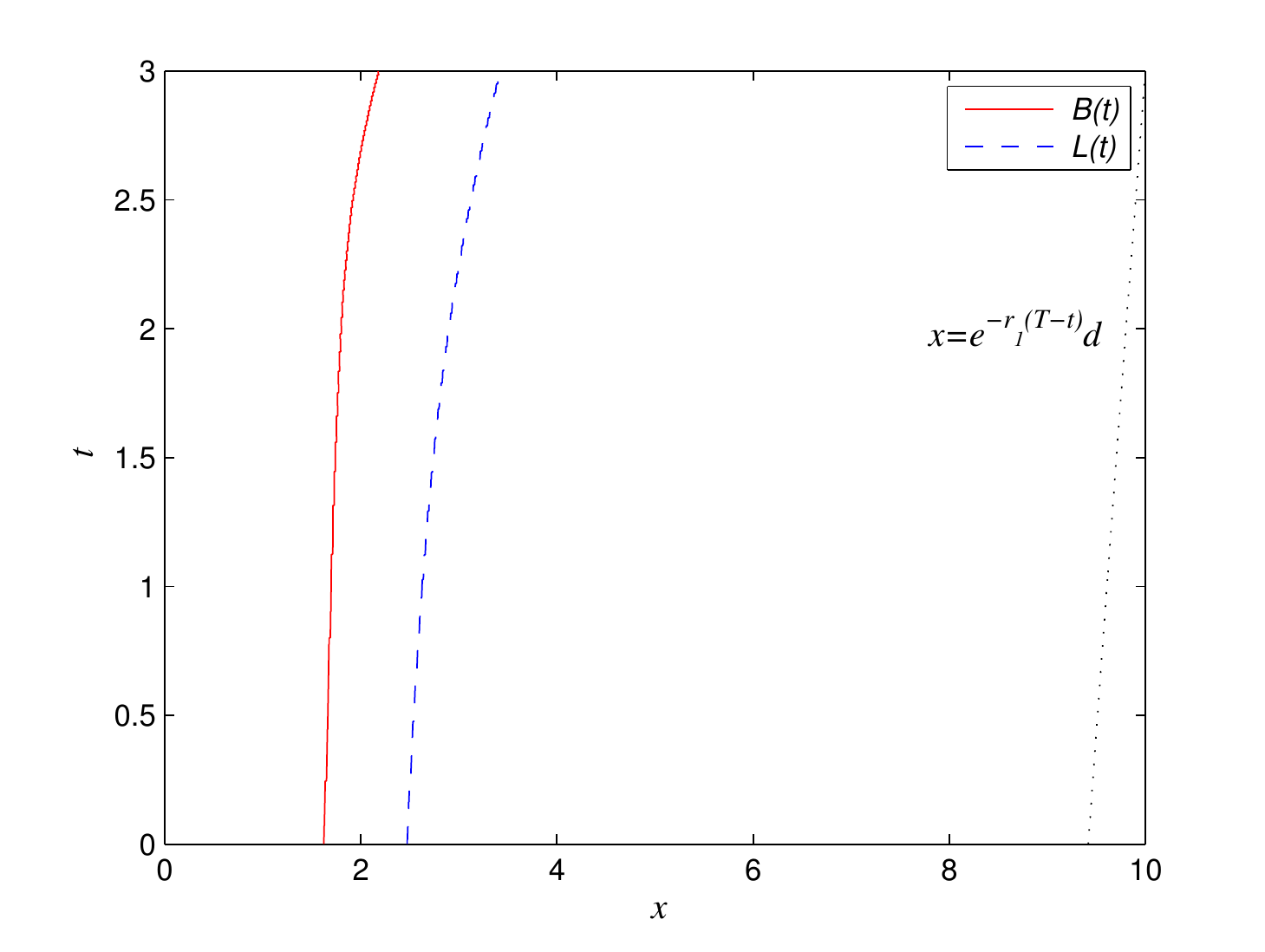} }
\subfigure[Free boundaries $B(\cdot)$ and $L(\cdot)$ when $\si^2=0.20$]
{\includegraphics[width=0.47\linewidth]{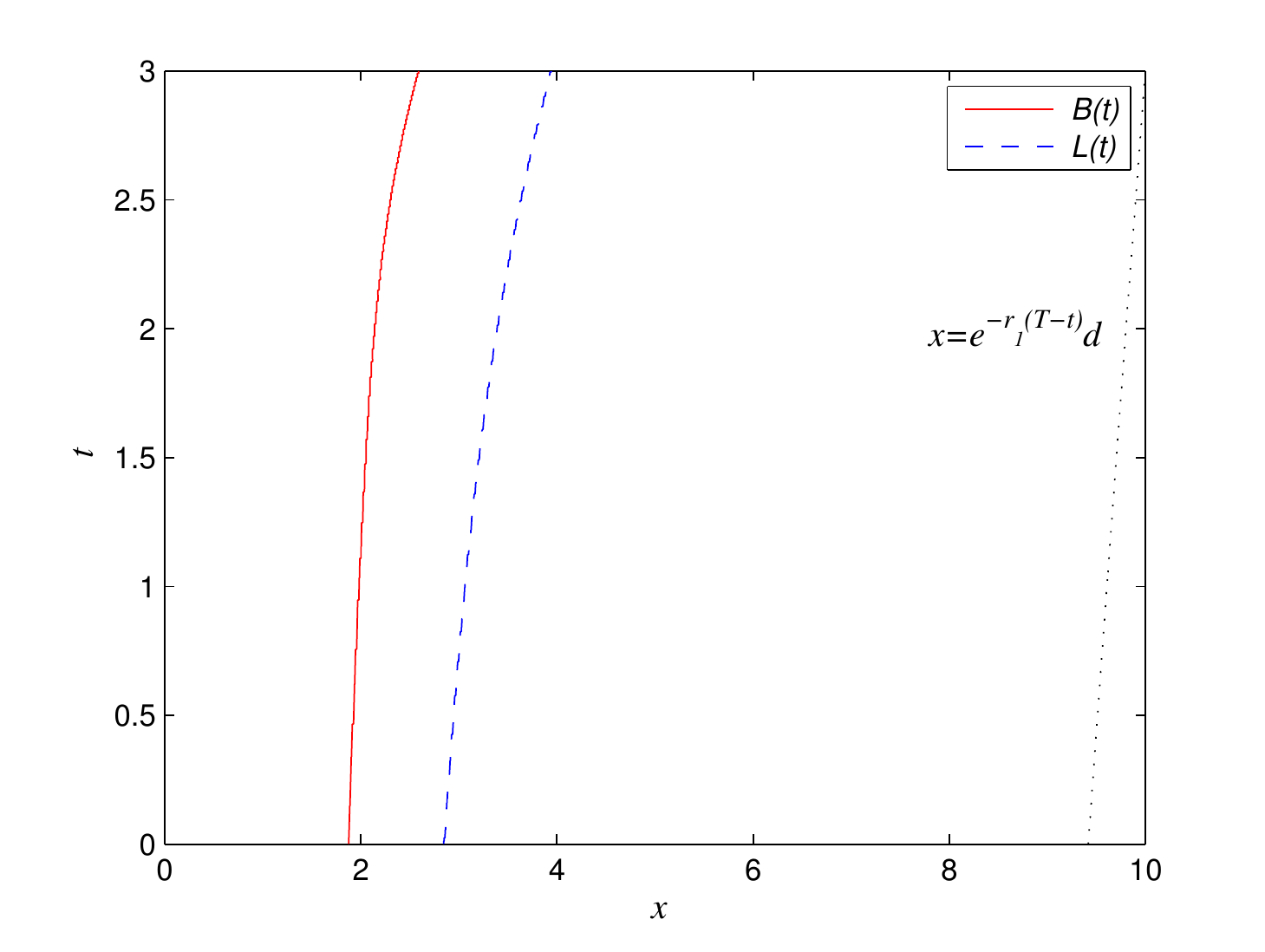} }
\subfigure[Free boundaries $B(\cdot)$ and $L(\cdot)$ when $\si^2=0.15$]
{\includegraphics[width=0.47\linewidth]{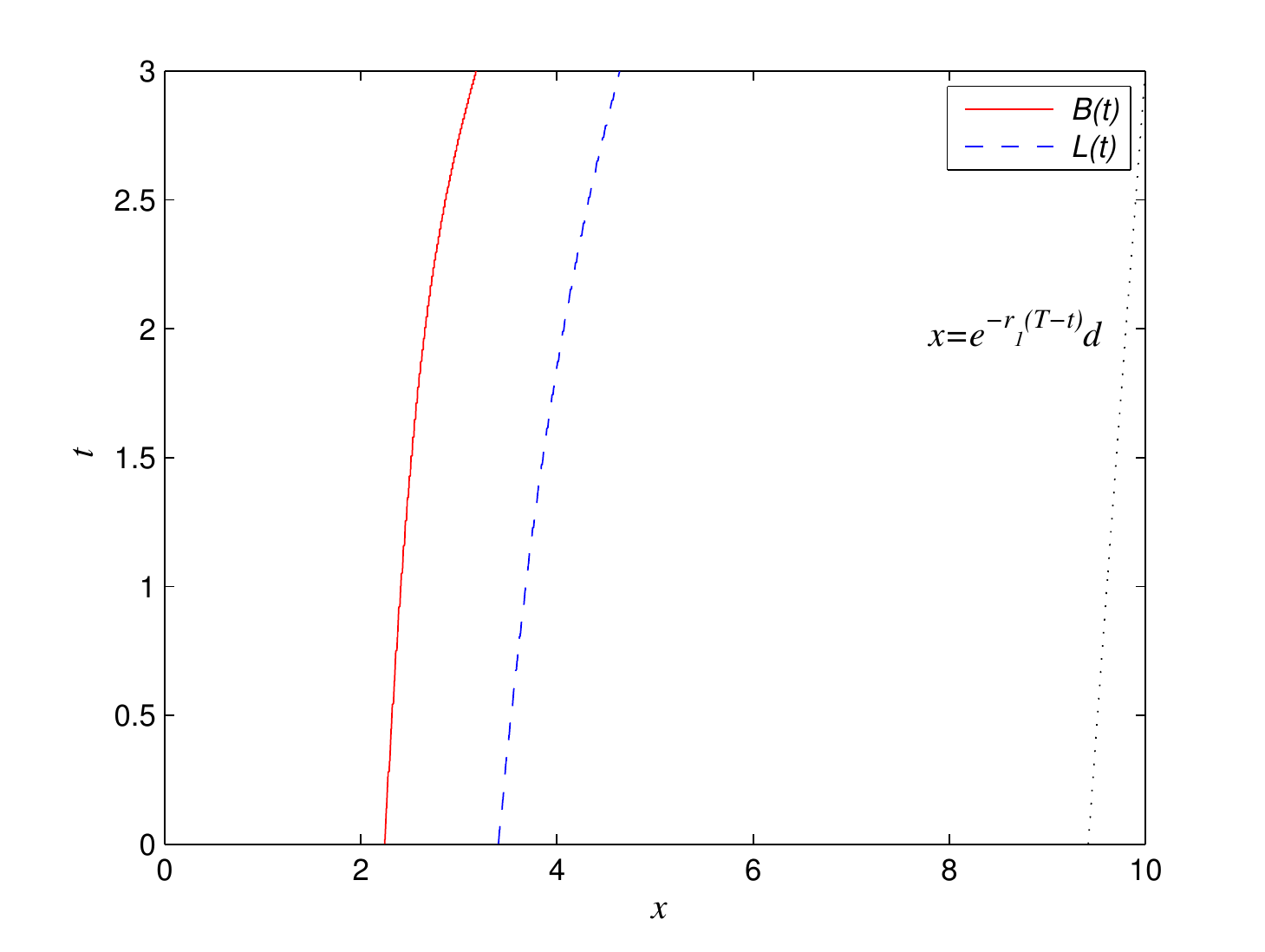} }
\end{center}
\caption{The borrowing and saving boundaries $B(\cdot)$ and $L(\cdot)$ under various volatility rates $\si$.}\label{fig_FB3}
\end{figure}
Figure \ref{fig_FB3} displays the borrowing and saving boundaries under various volatility rates $\si$, where $r_1=0.02$, $r_2=0.08$, $\mu=0.15$ and $d=10$ are the same as Figure \ref{fig_FB2}. Similar to Figure \ref{fig_FB2}, the borrowing and saving boundaries move to the right as $\si$ decreases. This is not surprising since no matter increasing $\mu$ or decreasing $\si$ will lead to an increment in the Shape ratio of the stock, making the stock more attractive to the investor. Therefore, Borrowing-money Region $\fB$ expands, and Saving-money Region $\fS$ shrinks, that is, one tends to borrow money to invest in the stock and is less likely to sell the stock to earn the saving rate. However, opposite to Figure \ref{fig_FB2}, All-in-stock Region $\fM$ expends here. We think, as $\si$ decreases, the uncertainty becomes less, hence less frequently tradings are enough to achieve the goal. In other words, one has to trade more frequently to control the risk if $\si$ increases.

\begin{figure}[H]
\begin{center}
\subfigure[Optimal proportion, $\pi(x,t)/x$, against $x$]
{\includegraphics[width=0.47\linewidth]{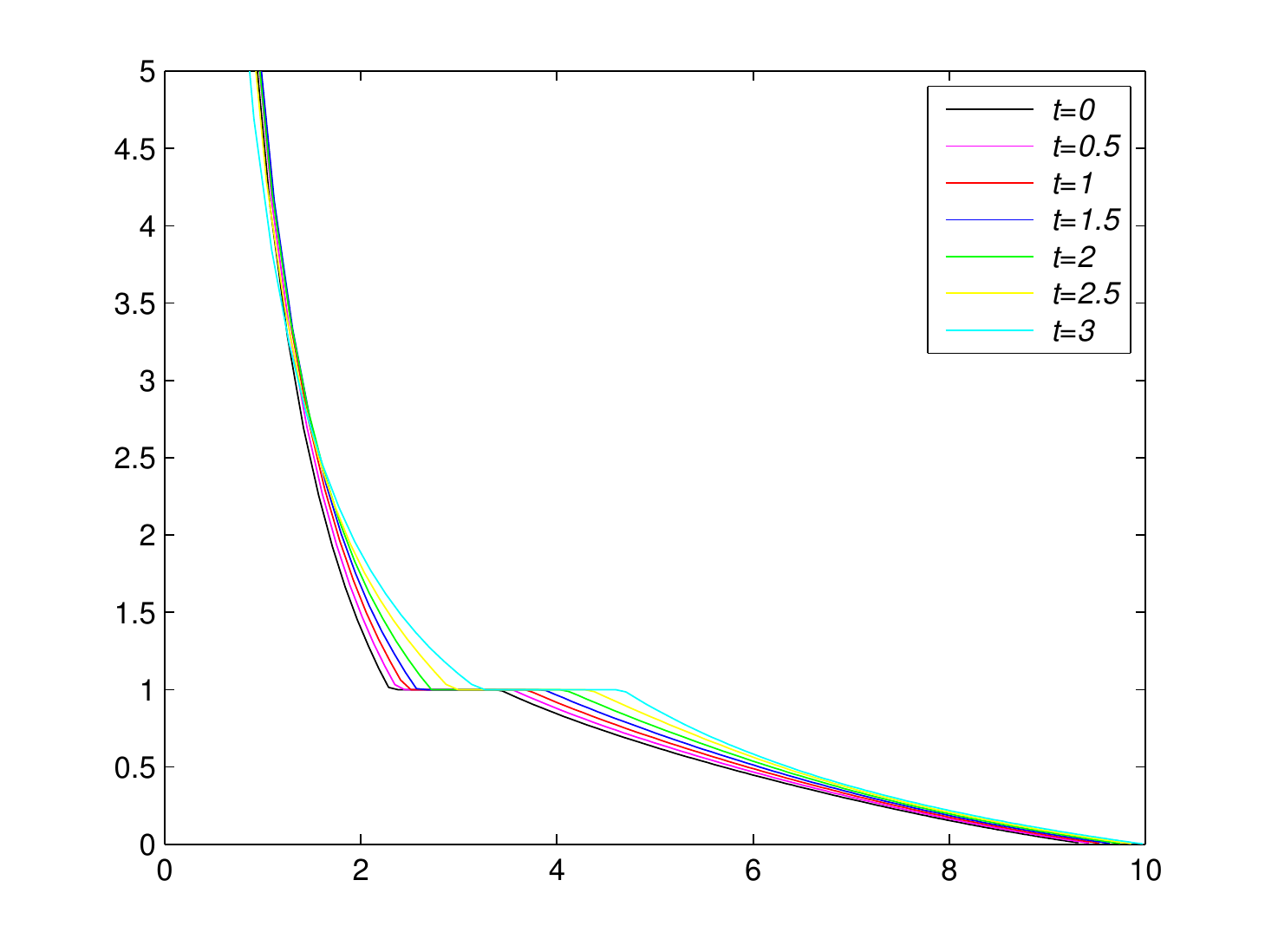} }\\
\subfigure[Optimal portfolio, $\pi(x,0)$, against $x$]
{\includegraphics[width=0.47\linewidth]{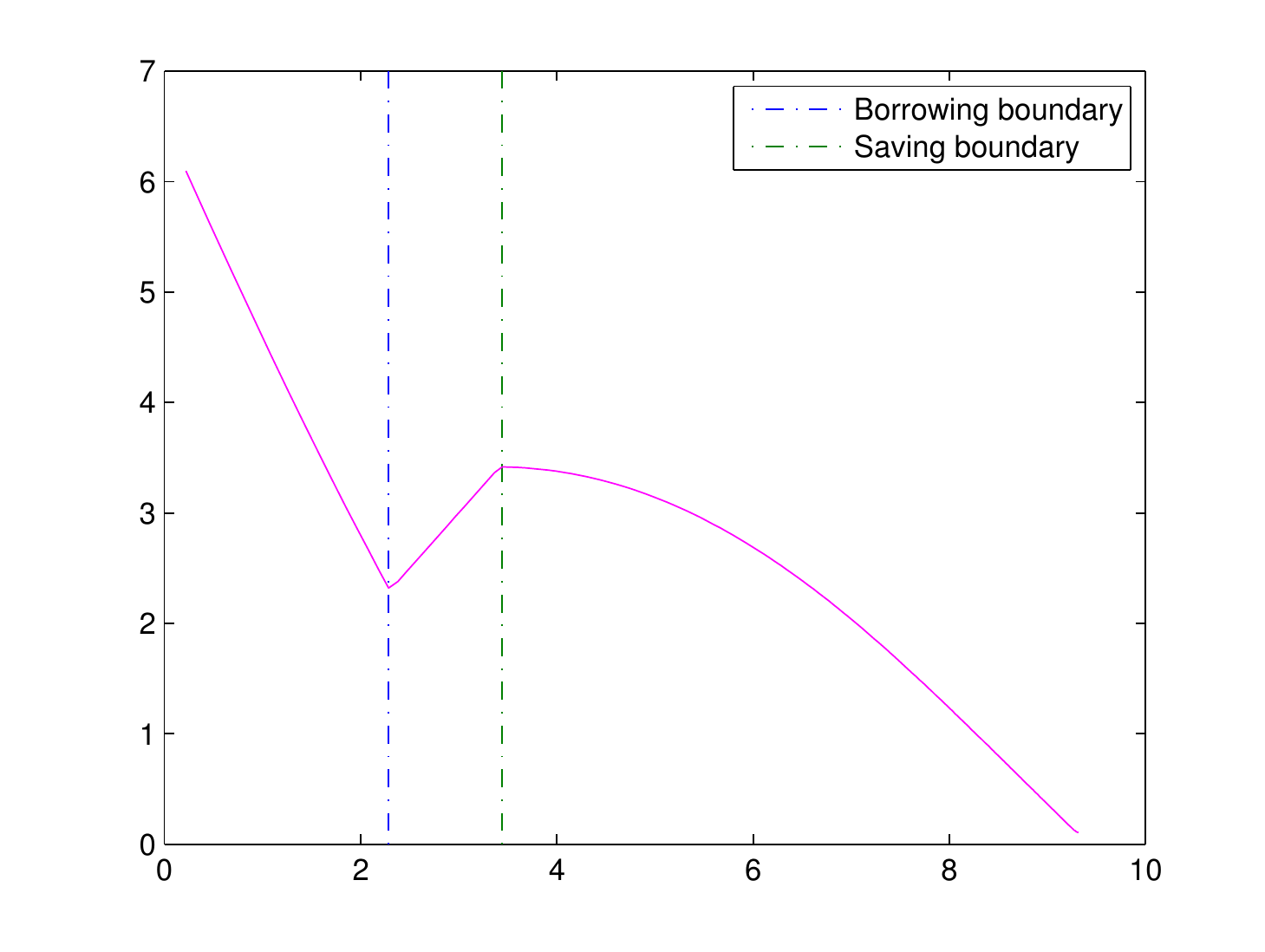} }
\subfigure[Optimal portfolio, $\pi(x,t)$, against $x$]
{\includegraphics[width=0.47\linewidth]{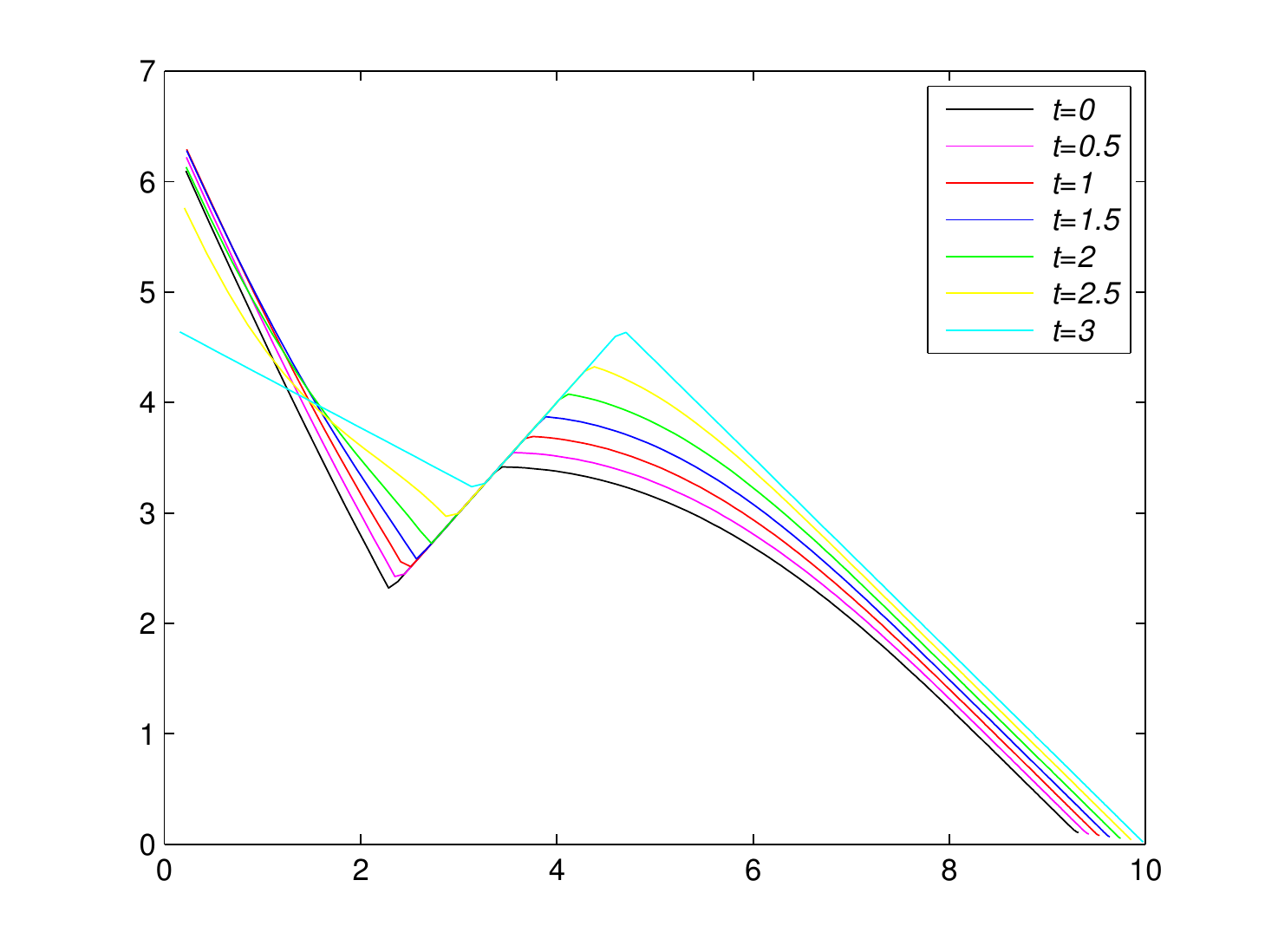} }
\caption{The optimal portfolio $\pi(x,t)$ under various times $t$.}\label{fig_pi}
\end{center}
\end{figure}

Figure \ref{fig_pi} displays the optimal proportion $\pi(x,t)/x$ and the optimal portfolio $\pi(x,t)$ under various times $t$, where $r_1=0.02$, $r_2=0.08$, $d=10$, $\mu=0.15$, $\si=0.15$ and $T=3$ are fixed. 
The above pictures show that $\pi(x,t)=x$ in All-in-stock Region $\fM$. Figure \ref{fig_pi}(b) and Figure \ref{fig_pi}(c) show that, as wealth $x$ increases to the discounted target $e^{-r_1(T-t)}d$, the net amount $\pi(x,t)$ invested in the stock is not globally monotone increasing or decreasing. Indeed, it is first decreasing in Borrowing-money Region $\fB$, then increasing All-in-stock Region $\fM$, and finally decreasing to zero in Saving-money Region $\fS$. However, we can observe from Figure \ref{fig_pi}(a) that the proportion invested in the stock is globally monotone decreasing to zero as $x$ increases. This is not surprising because one is expected to invest less proportion in the stock to reduce the risk as wealth approaches the discounted target.

Figure \ref{fig_pi}(a) also shows that the proportion invested in the stock becomes bigger in Saving-money Region $\fS$ as the time to maturity is closer. This is because one has to take a higher risk when there is less time available to achieve the target. By contrast, the behavior of the optimal proportion in Borrowing-money Region $\fB$ is quite complicate so that we cannot draw the same conclusion. But we can see that the slope of the optimal proportion in Borrowing-money Region $\fB$ becomes flatter as time approaches maturity, so the less time to maturity, the less sensitive to wealth position. Meanwhile, an opposite phenomenon is observed in Saving-money Region $\fS$.

The remaining part of this paper is devoted to the rigorous proofs for the theoretical results stated in Section \ref{sec:main}.

\section{Related Semi-linear PDE}\label{sec:Eq}

To study the fully nonlinear PDE \eqref{V_pb}, in this section, we transform it into a semi-linear PDE \eqref{w_pb} that satisfies the usual structural conditions by a heuristic argument. Many a priori estimates of the solution will be used in this process. In the following Section \ref{sec:hjb}, we will rigorously prove the existence and uniqueness of the solution to the PDE \eqref{w_pb} as well as those prior estimates used, and finally we will construct a solution to the PDE \eqref{V_pb} from the solution to the PDE \eqref{w_pb} in Section \ref{sec:V_solu}.

Our argument in the rest part of this section is intuitive and it will lead to a more tractable PDE \eqref{w_pb} which will service as our starting point of the theoretical treatment in the next section.

Our subsequent argument is based on the following hypotheses
\begin{align}\label{Vx_prb}
V_x < 0,\quad V_{xx} > 0,\quad (x,t)\in \TQ,
\end{align}
and
\begin{align}\label{V_lim_prb}
\lim\limits_{x\rightarrow e^{-r_1 (T-t)}d-}V_x=0,\quad \lim\limits_{x\rightarrow-\infty}V_x=-\infty,\quad t\in[0,T].\end{align}
These hypotheses will be eventually proved in Theorem \ref{theo:V}.

In order to solve the optimization problem in the HJB equation \eqref{V_pb}, write
\[
H(\pi):=\frac{1}{2}\si^2\pi^2V_{xx}+\big((r_1\chi_{x>\pi}+r_2\chi_{x<\pi})(x-\pi)+\mu\pi\big)V_x,\]
and
$$
\pi^*_i:=-a_i\frac{V_{x}}{V_{xx}},\quad a_i:=\frac{\mu-r_i}{\si^2},\quad i=1,2.
$$
By our assumption \eqref{murr}, $a_1>a_2>0$, so that it follows from \eqref{Vx_prb} that $\pi^*_1>\pi^*_2$. Consequently, only three possible scenarios can happen, which are demonstrated in Figure \ref{strategies}.

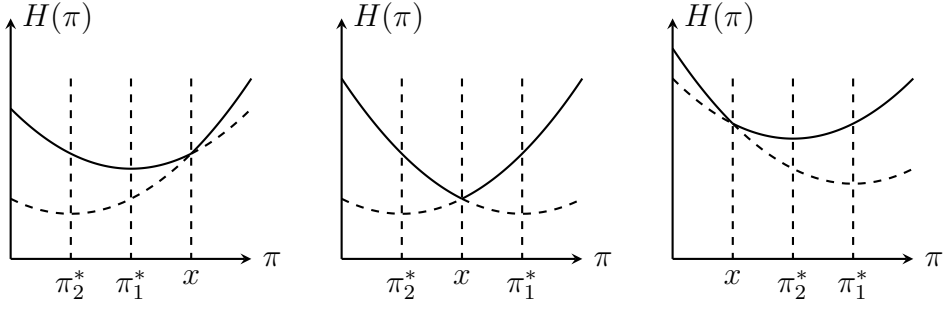
\begin{figure}[H]
\begin{center}
\begin{tikzpicture}[thick, scale=0.8]
\draw[eaxis] (0,0)--(4,0) node[right] {$\pi$};
\draw[eaxis] (0,0)--(0,4) node[right] {$H(\pi)$};
\draw[ dashed,black,domain=3:4] plot(\x,{(1/2*(\x)^2-1*(3-\x)-3*\x)/2+4});
\draw[elegant,black,domain=0:3] plot(\x,{(1/2*(\x)^2-1*(3-\x)-3*\x)/2+4});
\draw[elegant,black,domain=3:4] plot(\x,{(1/2*(\x)^2-2*(3-\x)-3*\x)/2+4});
\draw[dashed,black,domain=0:3] plot(\x,{(1/2*(\x)^2-2*(3-\x)-3*\x)/2+4});
\draw[-,dashed,] (3,3)--(3,0) node[below] {$x$};
\draw[-,dashed,] (2,3)--(2,0) node[below] {$\pi^*_1$};
\draw[-,dashed,] (1,3)--(1,0) node[below] {$\pi^*_2$}; 
\draw[eaxis] (0+5.5,0)--(4+5.5,0) node[right] {$\pi$};
\draw[eaxis] (0+5.5,0)--(0+5.5,4) node[right] {$H(\pi)$};
\draw[dashed,black,domain=2+5.5:4+5.5] plot(\x,{(1/2*(\x-5.5)^2-1*(2-(\x-5.5))-4*(\x-5.5))/2+4});
\draw[elegant,black,domain=0+5.5:2+5.5] plot(\x,{(1/2*(\x-5.5)^2-1*(2-(\x-5.5))-4*(\x-5.5))/2+4});
\draw[dashed,black,domain=0+5.5:2+5.5] plot(\x,{(1/2*(\x-5.5)^2-3*(2-(\x-5.5))-4*(\x-5.5))/2+4});;
\draw[elegant,black,domain=2+5.5:4+5.5] plot(\x,{(1/2*(\x-5.5)^2-3*(2-(\x-5.5))-4*(\x-5.5))/2+4});;
\draw[-,dashed,] (2+5.5,3)--(2+5.5,0) node[below] {$x$};
\draw[-,dashed,] (3+5.5,3)--(3+5.5,0) node[below] {$\pi^*_1$};
\draw[-,dashed,] (1+5.5,3)--(1+5.5,0) node[below] {$\pi^*_2$}; 
\draw[eaxis] (0+11,0)--(4+11,0) node[right] {$\pi$};
\draw[eaxis] (0+11,0)--(0+11,4) node[right] {$H(\pi)$};
\draw[dashed,black,domain=1+11:4+11] plot(\x,{(1/2*(\x-11)^2-1*(1-(\x-11))-4*(\x-11))/2+4});
\draw[elegant,black,domain=0+11:1+11] plot(\x,{(1/2*(\x-11)^2-1*(1-(\x-11))-4*(\x-11))/2+4});
\draw[dashed,black,domain=0+11:1+11] plot(\x,{(1/2*(\x-11)^2-2*(1-(\x-11))-4*(\x-11))/2+4});;
\draw[elegant,black,domain=1+11:4+11] plot(\x,{(1/2*(\x-11)^2-2*(1-(\x-11))-4*(\x-11))/2+4});;
\draw[-,dashed,] (1+11,3)--(1+11,0) node[below] {$x$};
\draw[-,dashed,] (3+11,3)--(3+11,0) node[below] {$\pi^*_1$};
\draw[-,dashed,] (2+11,3)--(2+11,0) node[below] {$\pi^*_2$};
\end{tikzpicture}
\end{center}
\caption{Three possible cases.} \label{strategies}
\end{figure}

\noindent
Correspondingly, we have
\[
\argmin_{\pi} H(\pi)=
\begin{cases}
-a_1\frac{V_{x}}{V_{xx}}, &-a_1\frac{V_{x}}{V_{xx}}<x,\\[3mm]
x, &-a_2\frac{V_{x}}{V_{xx}}\leq x\leq-a_1\frac{V_{x}}{V_{xx}},\\[3mm]
-a_2\frac{V_{x}}{V_{xx}}, &-a_2\frac{V_{x}}{V_{xx}}>x.
\end{cases}
\]
Inserting this into \eqref{V_pb}, we get
\begin{align}\label{Three_eq}
\begin{cases}
-V_t+\frac{\si^2 a_1^2}{2}\frac{V_x^2}{V_{xx}}+(\si^2a_1-\mu)xV_x=0, \quad &-a_1\frac{V_{x}}{V_{xx}}<x,\\[3mm]
-V_t-\frac{\si^2}{2}x^2V_{xx}-\mu xV_x=0, &-a_2\frac{V_{x}}{V_{xx}}\leq x\leq-a_1\frac{V_{x}}{V_{xx}},\\[3mm]
-V_t+\frac{\si^2 a_2^2}{2}\frac{V_x^2}{V_{xx}}+(\si^2a_2-\mu)xV_x=0, &-a_2\frac{V_{x}}{V_{xx}}>x,\\[3mm]
V(e^{-r_1(T-t)} d,t)=0,& 0\leq t<T,\\[3mm]
V(x,T)=(x-d)^2, &x<d.
\end{cases}
\end{align}
This is a fully nonlinear free boundary problem that does not satisfy the general structural conditions of nonlinear parabolic equation, so it is hard to apply the existing results to study it directly. We need to rewrite it in a more tractable form.

By \eqref{Vx_prb}, $V$ should be a convex function. This motives us to apply the dual transformation (see, e.g., Pham \cite{Ph09}) to simplify \eqref{Three_eq}. To this end, let
\begin{align*}
v(y,t):=\inf\limits_{x< e^{-r_1(T-t)}d}\big(V(x,t)+xy\big),\quad y>0,\;\;0\leq t\leq T.\end{align*}
Thanks to the hypotheses \eqref{Vx_prb} and \eqref{V_lim_prb}, we get $v(0,t)=0$. 
By \eqref{V_lim_prb}, for each fixed $t\in [0,T]$, the optimal $x$ corresponding to $y$ is
\[
x=x(y,t):=V_x^{-1}(\cdot,t)(-y),\quad y>0,\]
where $V_x^{-1}(\cdot,t)$ is the inverse of $V_x(\cdot,t)$. This gives the following correspondence between $v(y,t)$ and $V(x,t)$,
\begin{align*}
v(y,t)&=V(x(y,t),t)+x(y,t)y,\\
v_y(y,t)&=V_x(x(y,t),t)x_y(y,t)+y x_y(y,t)+x(y,t)=x(y,t),\\
v_{yy}(y,t)&=x_y(y,t)=\frac{-1}{V_{xx}(x(y,t),t)},\\
v_t(y,t)&=V_t(x(y,t),t)+V_x(x(y,t),t)x_t(y,t)+y x_t(y,t)=V_t(x(y,t),t).
\end{align*}
It then follows from \eqref{Three_eq} that
\begin{align}\label{v_eq}
\begin{cases}
-v_t-\frac{1}{2}\si^2 a_1^2y^2v_{yy}-(\si^2a_1-\mu) yv_y=0, &-\frac{v_y}{v_{yy}}>a_1y,\\[3mm]
-v_t+\frac{\si^2}{2}\frac{v_y^2}{v_{yy}}+\mu yv_y=0, & a_2y\leq-\frac{v_y}{v_{yy}}\leq a_1y,\\[3mm]
-v_t-\frac{1}{2}\si^2 a_2^2y^2v_{yy}-(\si^2a_2-\mu) yv_y=0, &-\frac{v_y}{v_{yy}}<a_2y,\\[3mm]
v(0,t)=0, & 0\leq t<T,\\[3mm]
v(y,T)=-\frac{1}{4}y^2+d y, & y>0.
\end{cases}
\end{align}
Now we introduce
$$
u:=-v_y.
$$
After differentiating \eqref{v_eq} w.r.t. $y$, we obtain an equation for $u$:
\begin{align}\label{u_eq}
\begin{cases}
-u_t-\frac{1}{2}\si^2 a_1^2y^2u_{yy}+(-\si^2a_1^2-\si^2a_1+\mu)yu_y+(\mu-\si^2a_1)u=0, &-\frac{u}{u_y}>a_1y,\\[3mm]
-u_t-\frac{\si^2}{2}\big(\frac{u}{u_y}\big)^2u_{yy}+\si^2 u+\mu yu_y+\mu u=0, & a_2y\leq-\frac{u}{u_y}\leq a_1y,\\[3mm]
-u_t-\frac{1}{2}\si^2 a_2^2y^2u_{yy}+(-\si^2a_2^2-\si^2a_2+\mu)yu_y+(\mu-\si^2a_2)u=0, &-\frac{u}{u_y}<a_2y,\\[3mm]
u(y,T)=\frac{1}{2}y-d, &y>0.
\end{cases}
\end{align}
Making a transformation $u(y,t)=w(z,s)$ for $s=T-t$, $z=\ln y$, we get
\[
u_t=-w_s,\quad u_y=w_z\frac{1}{y},\quad u_{yy}=\big(w_{zz}-w_z\big)\frac{1}{y^2},\]
so that \eqref{u_eq} becomes
\begin{align}\label{three_eq}
\begin{cases}
w_s-\frac{1}{2}\si^2 a_1^2w_{zz}+\big(\mu-\frac{1}{2}\si^2a_1^2-\si^2a_1\big)w_z+(\mu-\si^2a_1)w=0, &-\frac{w}{w_z}>a_1,\\[3mm]
w_s-\frac{\si^2}{2}\Big(\frac{w}{w_z}\Big)^2\big(w_{zz}-w_z\big)+\si^2 w+\mu w_z+\mu w=0, & a_2\leq-\frac{w}{w_z}\leq a_1,\\[3mm]
w_s-\frac{1}{2}\si^2 a_2^2w_{zz}+\big(\mu-\frac{1}{2}\si^2a_2^2-\si^2a_2\big)w_z+(\mu-\si^2a_2)w=0, &-\frac{w}{w_z}<a_2,\\[3mm]
w(z,0)=\frac{1}{2}e^z-d, &z\in\R.
\end{cases}
\end{align}

For convenience, we define a function
\[
A(\xi):=\min\big\{\max\{a_2,-\xi\}, a_1\big\}=
\begin{cases}
a_1, &-\xi>a_1,\\[2mm]
-\xi, & a_2\leq-\xi\leq a_1,\\[2mm]
a_2, &-\xi<a_2.
\end{cases}
\]
It is clearly a 1-Lipschitz-continuous bounded decreasing function. Now we can rewrite \eqref{three_eq} in a compact form as an initial value problem
\begin{align}\label{w_pb}
\begin{cases}
w_s-{\cal T}w=0 \quad \hbox{in}\quad Q_T:=\R\times(0,T],\\[3mm]
w(z,0)=\frac{1}{2}e^z-d,\quad z\in \R,
\end{cases}
\end{align}
where
\[
{\cal T}w:=\frac{1}{2}\si^2A^2\Big(\frac{w}{w_z}\Big)w_{zz}
-\Big[\mu-\frac{1}{2}\si^2A^2\Big(\frac{w}{w_z}\Big)-\si^2A\Big(\frac{w}{w_z}\Big)\Big]w_z
-\Big[\mu-\si^2A\Big(\frac{w}{w_z}\Big)\Big]w.\]
By definition we have $$0<a_2\leq A(\cdot)\leq a_1,$$ so \eqref{w_pb} is a semi-linear parabolic PDE, which satisfies the usual structural conditions. We will study its solvability and properties in the next section.

\section{Solvability of the PDE \eqref{w_pb}}\label{sec:hjb}

We have got a semi-linear parabolic PDE \eqref{w_pb} through an intuitive argument in the previous section. From now on, we do a rigorous analysis and focus on the solvability and properties of the PDE \eqref{w_pb} in this section.

We first introduce several constants that will be used throughout the paper: 
\bee
\theta_1:=\si^2 a_1^2-2r_1,\quad \theta_2:=\si^2 a_2^2-2r_2,\eee
and
\bee
k :=\max\Big\{2\si^2 a_1+4\si^2 a_1^2,\;\theta_1\Big\}, \quad
\kappa:=2\mu+\si^2 (3a_1+1)(a_1+1).\eee

Our main theoretical result is stated as follows. 
\begin{theorem}\label{theo:w}
There exists a solution $w\in C^{2+\al,1+\frac{\al}{2}} \big(\R\times [0,T]\big)$ (for some $\al\in (0,1)$) to the PDE \eqref{w_pb} such that
\begin{align}\label{w_b}
\frac{1}{2} e^{\theta_2 s} e^z-e^{-r_1 s} d\leq & w \leq \frac{1}{2} e^{\theta_1 s} e^z-e^{-r_2 s} d,\end{align}
and
\begin{align}\label{wz_b}
\frac{1}{2}e^{-\kappa s} e^z\leq & w_z\leq \frac{1}{2}e^{k s}e^z.\end{align}
\end{theorem}

\begin{proof}
The proof, which needs some results of Sobolev space and a priori estimation method of parabolic equations, is cumbersome, so we put it in Appendix \ref{sec:w_exist}.
\end{proof}

In the rest part of this paper we fix a solution $w$ as in \thmref{theo:w}. Based on it, we will construct solutions to \eqref{u_eq}, \eqref{v_eq}, \eqref{Three_eq} and \eqref{V_pb} in the following subsections.
In particular, Theorem \ref{veri} will ensure such $w$ is indeed unique.

\begin{remark}
The exact values of $\theta_1$, $\theta_2$, $k$ and $\kappa$ in \thmref{theo:w} are not important. We just need to make sure that $w$ and $w_z$ are growth exponentially in $z$, which will suffice to ensure such solution $w$ to \eqref{w_pb} is unique.
\end{remark}

\subsection{Free Boundaries of \eqref{w_pb}}\label{sec:free boundary}

In order to study the properties of \eqref{w_pb}, we define three sets
\begin{align*}
{\cal B}&:=\Big\{(z,s)\in Q_T\;\Big|\;-\frac{w}{w_z}< a_2\Big\}, \\
{\cal N}&:=\Big\{(z,s)\in Q_T\;\Big|\;a_2\leq-\frac{w}{w_z}\leq a_1\Big\}, \\
{\cal S}&:=\Big\{(z,s)\in Q_T\;\Big|\;-\frac{w}{w_z}> a_1\Big\},
\end{align*}
and define two free boundaries
\begin{align}\label{boundaryb}
b(s)&:=\;\sup\Big\{z\in \R\;\Big|\;-\frac{w}{w_z}(z,s)\geq a_2\Big\}=\sup\big\{z\in \R\mid (w+a_2 w_z)(z,s)\leq 0\big\},\quad s\in(0,T],\\
l(s)&:=\inf\Big\{z\in \R\;\Big|\;-\frac{w}{w_z}(z,s)\leq a_1\Big\}=\inf\big\{z\in \R\mid (w+a_1 w_z)(z,s)\geq 0\big\},\quad s\in(0,T],\label{boundaryl}
\end{align}
where we used $w_z>0$ to get the second expressions in above. They will be used to study the properties of the optimal portfolio for our original problem \eqref{value}. Because $a_1>a_2$ and $w_z>0$, we see $b(s)>l(s)$ for all $s\in(0,T]$.

We have the following estimates for the two boundaries $b(\cdot)$ and $l(\cdot)$. 
\begin{lemma}\label{lem:b_ub}
For any $s\in (0,T]$,
\begin{align}\label{b_ub}
b(s)< \ln(2d)-(r_1+\theta_2)s.
\end{align}
\end{lemma}
\begin{proof}
If $z\geq \ln(2d)-(r_1+\theta_2)s$, then
by $a_2>0$ and the lower bounds in \eqref{w_b} and \eqref{wz_b} we have
\[ w(z,s)+a_2w_z>\frac{1}{2} e^{\theta_2 s} e^z-e^{-r_1 s} d\geq 0\]
implying \eqref{b_ub}.
\end{proof}

\begin{lemma}\label{lem:l_lb}
For any $s\in (0,T]$,
\begin{align}\label{l_lb}
l(s)\geq \ln(2d)-\ln(a_1+e^{(\theta_2-k)s})-(r_1+k)s.\end{align}
\end{lemma}
\begin{proof}
If $z<\ln(2d)-\ln(a_1+e^{(\theta_1-k)s})-(r_2+k)s$, by the upper bounds in \eqref{w_b} and \eqref{wz_b} we have
\[w+a_1w_z\leq \frac{1}{2} e^{\theta_1 s} e^z-e^{-r_2 s} d+a_1\frac{1}{2}e^{k s}e^z
=\frac{1}{2}e^{k s}e^z\big(e^{(\theta_1-k)s}-2e^{-(r_2+k) s-z} d+a_1\big)<0,\]
which implies \eqref{l_lb}.
\end{proof}

Define two functions
$$
I:=w+a_2 w_z
$$
and
$$
f(s):=\sup\{z\in \R\mid I(z,s)< 0\},\quad s\in (0,T].
$$
By definition, we have $f(s)\leq b(s)$.

\begin{lemma}\label{lem:b+-}
If $f_*(s_0-)< f^*(s_0)$ for some $s_0\in(0,T]$, then
\begin{align}\label{b+-}
I(z,s_0)=0,\quad\forall z\in(f_*(s_0-),f^*(s_0)),\end{align}
where $f_*(s_0-):=\liminf\limits_{s\rt s_0-}f(s)$ and $f^*(s_0):=\limsup\limits_{s\rightarrow s_0}f(s)$.
\end{lemma}
\begin{proof}
By the continuity of $I$ and the definition of $f(s)$, we have
\begin{align}\label{wp1}
I(z,s_0)\geq0,\quad z\geq f_*(s_0-).\end{align}
If \eqref{b+-} were not true, then there would exits $z_0\in(f_*(s_0-),f^*(s_0))$ such that $I(z_0,s_0)>0$. Owing to the continuity, we would have
\be\label{wp2}
I(z_0,s)>0,\;s\in (s_0-\ep,s_0+\ep)
\ee
for sufficiently small $\ep>0$. Since $z_0>f_*(s_0-)$, we could suppose $z_0>f(s_0-\ep)$ so that
\be\label{wp3}
I(z,s_0-\ep)\geq 0,\quad z>z_0.\ee

Now we would prove $I> 0$ in $${\cal D}:=(z_0,+\infty)\times(s_0-\ep,s_0+\ep).$$ 
Indeed, suppose $\psi$ is the unique solution to
\begin{align}\label{psi}
\begin{cases}
\psi_s-\frac{1}{2}\si^2 a_2^2\psi_{zz}+\big(\mu-\frac{1}{2}\si^2a_2^2-\si^2a_2\big)\psi_z+(\mu-\si^2a_2)\psi=0 \quad \hbox{in}\quad {\cal D},\\[3mm] \big(\psi+a_2\psi_z\big)(z_0,s)=I(z_0,s),\quad s\in(s_0-\ep,s_0+\ep),\\[3mm]
\psi(z,s_0-\ep)=w(z,s_0-\ep),\quad z> z_0.
\end{cases}
\end{align}
under the exponential growth conditions on $\psi$ and $\psi_z$.
Differentiating the equation in \eqref{psi} w.r.t. $z$, we have
\[
\psi_{zs}-\tfrac{1}{2}\si^2 a_2^2\psi_{zzz}+\big(\mu-\tfrac{1}{2}\si^2a_2^2-\si^2a_2\big)\psi_{zz}+(\mu-\si^2a_2)\psi_z=0 \quad \hbox{in}\quad {\cal D},\]
So $\Psi=\psi+a_2\psi_z$ satisfies
\begin{align*}
\begin{cases}
\Psi_s-\frac{1}{2}\si^2 a_2^2\Psi_{zz}+\big(\mu-\frac{1}{2}\si^2a_2^2-\si^2a_2\big)\Psi_z+(\mu-\si^2a_2)\Psi=0 \quad \hbox{in}\quad {\cal D},\\[3mm] \Psi(z_0,s)=I(z_0,s),\quad s\in(s_0-\ep,s_0+\ep),\\[3mm]
\Psi(z,s_0-\ep)=I(z,s_0-\ep),\quad z> z_0.
\end{cases}
\end{align*}
Using \eqref{wp2} and \eqref{wp3}, by the strong maximum principle, we have $\Psi> 0$ in ${\cal D}$.
Define a function
\begin{align*}
\Gamma(x,y)=
\left\{
\begin{array}{ll}
a_2, & x+a_2 y\geq 0, \\[3mm]
a_1, & x+a_2 y< 0 \hbox{\quad and\quad } x+a_1 y\leq 0, \\[3mm]
-\frac{x}{y},\;\;& x+a_2 y< 0 \hbox{\quad and\quad } x+a_1 y> 0.
\end{array}
\right.
\end{align*}
Then \eqref{psi} can be rewritten as
\begin{align*}
\begin{cases}
\psi_s-\frac{1}{2}\si^2 \Gamma^2(\psi,\psi_z)\psi_{zz}+\big(\mu-\frac{1}{2}\si^2\Gamma^2(\psi,\psi_z) -\si^2\Gamma(\psi,\psi_z)\big)\psi_z+\big(\mu-\si^2\Gamma(\psi,\psi_z)\big)\psi=0,\\[3mm] \big(\psi+a_2\psi_z\big)(z_0,s)=\big( w+a_2w_z\big)(z_0,s),\quad s\in(s_0-\ep,s_0+\ep)\\[3mm]
\psi(z,s_0-\ep)=w(z,s_0-\ep),\quad z> z_0.
\end{cases}
\end{align*}
Since $w_z>0$ by \eqref{wz_b}, we have $\Gamma(w,w_z)=A(w/w_z)$. Hence $\psi=w$ also satisfies the above system, by the uniqueness of its solution, we conclude that $\psi=w$ in ${\cal D}$. Consequently, $I=\Psi>0 $ in ${\cal D}$. But, by the definition of $f(s_0)$, we would have $f(s)\leq z_0$ for $s\in(s_0-\ep,s_0+\ep)$ so that 
$$f^*(s_0)=\limsup\limits_{s\rightarrow s_0}f(s)\leq z_0,$$ 
contradicting to $z_0\in(f_*(s_0-),f^*(s_0))$.
\end{proof}

\begin{lemma}\label{lem:<f}
Given $s\in (0,T]$, we have
\begin{align}\label{<f}
I(z,s)\leq 0,\quad \forall \; z\leq f(s).\end{align}
\end{lemma}
\begin{proof}
Denote $${\cal C}:=\{(z,s)\mid z\leq f(s),\;s\in(0,T]\}.$$ If \eqref{<f} were not true, i.e. $I$ would take positive values in ${\cal C}$. Since $\limsup\limits_{z\rt-\infty}I(z,s)<0$ for any $s\in(0,T]$ by \eqref{w_b} and \eqref{wz_b}, there would exist $(z_0,s_0)\in\overline{{\cal C}}$ such that $I(z_0,s_0)=\max\limits_{(z,s)\in\overline{{\cal C}}}I(z,s)>0.$ Note that $(z_0,s_0)\in\overline{{\cal C}}$ implies 
$$z_0\leq f^*(s_0)=\limsup\limits_{s\rightarrow s_0}f(s).$$ 
By \lemref{lem:b+-} we would have $$z_0< f_*(s_0-)=\liminf\limits_{s\rt s_0-}f(s).$$ Therefore, $I>0$ in $${\cal D}:=(z_0-\ep,z_0+\ep)\times(s_0-\ep,s_0)\subset {\cal C}$$ for sufficiently small $\ep>0$. Then $A(w/w_z)=a_2$ in ${\cal D}$, so $I$ would satisfy a linear equation in ${\cal D}$. However, as $(z_0,s_0)$ is the maximum point of $I$ in ${\cal D}$, it is impossible by the maximum principle.
\end{proof}

By \lemref{lem:<f}, we see $I(z,s_0)\geq0$ for any $z> z_0$ if $I(z_0,s_0)>0$. We continue to prove the following stronger conclusion.
\begin{lemma}\label{lim:>f}
Given $s_0\in (0,T]$, we have $I(z,s_0)>0$ for any $z> z_0$ if $I(z_0,s_0)>0$.
\end{lemma}
\begin{proof}
By the continuity of $I$, there exists $0<\ep<s_0$ such that
$$
I(z_0,s)>0,\quad s\in (s_0-\ep,s_0].
$$
By \lemref{lem:<f} we further have
$$
I(z,s)\geq 0,\quad (z,s)\in (z_0,+\infty)\times(s_0-\ep,s_0]
$$
and $I$ satisfies a linear equation in $(z_0,+\infty)\times(s_0-\ep,s_0]$. By the strong maximum principle we conclude $I>0$ in $(z_0,+\infty)\times(s_0-\ep,s_0]$.
\end{proof}

By \lemref{lim:>f} and the definition \eqref{boundaryb}, we conclude
\begin{lemma}\label{lem:B}
We have
\begin{align}\label{B}
{\cal B}=\{(z,s)\mid z>b(s),\;s\in(0,T]\}\end{align}
with $$b(0+)=\ln(2d)-\ln(1+a_2).$$
\end{lemma}

Similarly, we can prove
\begin{lemma}\label{lem:L}
We have
\begin{align}\label{L}
{\cal S}=\{(z,s)\mid z<l(s),\;s\in(0,T]\} \end{align}
with $$l(0+)=\ln(2d)-\ln(1+a_1).$$
\end{lemma}

Recall that $l(s)<b(s)$, so the above two lemmas imply
\begin{lemma}\label{lem:M}
We have
\begin{align*}
{\cal N}=\{(z,s)\mid l(s)\leq z\leq b(s),\;s\in(0,T]\}.\end{align*}
\end{lemma}

Next, we prove that the boundaries $b(\cdot)$ and $l(\cdot)$ are smooth when the coefficients meet certain conditions.

\begin{proposition}\label{prop:C1}
If the coefficients satisfy the following conditions
\begin{gather}
\label{CD0}
{\mu>0,}\\\label{CD1}
\si^2 a_2^2+r_1-2 r_2 \geq 0,\\\label{CD2}
1+\frac{2 r_1}{\mu-r_1}-2 a_1+a_2\geq 0,
\end{gather}
then the boundaries $b(\cdot)$, $l(\cdot)\in C^1((0,T])$.
\end{proposition}

To prove this conclusion, we need

\begin{lemma}\label{lem:ws>0}
Under the condition \eqref{CD1}, we have
\be\label{ws>0}
w_s+r_1 w\geq 0.
\ee
\end{lemma}
\begin{proof}
The condition \eqref{CD1} is equivalent to 
$$\theta_2=\si^2 a_2^2+2 \si^2 a_2-2 \mu\geq-r_1.$$ 
Denote $\varphi=e^{r_1 s} w$, by the first inequality in \eqref{w_b} we have $$\varphi(z,s)\geq \frac{1}{2}e^{(\theta_2+r_1)s}e^z-d\geq\frac{1}{2}e^z-d=\varphi(z,0),\quad (z,s)\in Q_T.$$ For any $\Delta s\in (0,T)$, let $\overline{\varphi}(z,s):=\varphi(z,s+\Delta s)$, then by above and the equation in \eqref{w_pb}, we have
\bee
\begin{cases}
\overline{\varphi}_s-r_1 \overline{\varphi}-{\cal T} \overline{\varphi}=0 \quad \hbox{in}\quad Q_{T-\Delta s},\\[3mm]
\overline{\varphi}(z,0)=\varphi(z,\Delta s)\geq \varphi(z,0),\quad z\in \R.
\end{cases}
\eee
By the comparison principle we have $\overline{\varphi}\geq \varphi$ in $Q_{T-\Delta s}$, which implies $\varphi_s\geq 0$, so $w_s+r_1 w\geq 0$.
\end{proof}

\begin{lemma}\label{lem:Az<}
Under the conditions \eqref{CD0}, \eqref{CD1} and \eqref{CD2},
\be\label{Az<}
\p_z\Big(\frac{w}{w_z}\Big)>0 \quad\rm{in}\quad {\cal N}.
\ee
\end{lemma}
\begin{proof}
Recall that $w<0$, $w_z>0$ and $a_2\leq A(\frac{w}{w_z})=-\frac{w}{w_z}\leq a_1$ in $ {\cal N}$, so
$$
A^2\Big(\frac{w}{w_z}\Big) w_{zz}=\Big(\frac{w}{w_z}\Big)^2 w_{zz}=w \Big(1-\p_z\Big(\frac{w}{w_z}\Big)\Big)
\quad\rm{in}\quad {\cal N}.
$$
By the equation in \eqref{w_pb} and \lemref{lem:ws>0} we have
\begin{align*}
\frac{\si^2}{2}w \Big(1-\p_z\Big(\frac{w}{w_z}\Big)\Big)
&=
w_s+\Big(\mu-\frac{1}{2}\si^2A^2\Big(\frac{w}{w_z}\Big)-\si^2A\Big(\frac{w}{w_z}\Big)\Big)w_z+\Big(\mu-\si^2A\Big(\frac{w}{w_z}\Big)\Big)w\\
&=
w_s-\frac{\mu } {A\Big(\frac{w}{w_z}\Big)}w+\frac{1}{2}\si^2A\Big(\frac{w}{w_z}\Big)w+\si^2w+\Big(\mu-\si^2A\Big(\frac{w}{w_z}\Big)\Big)w\\
&=
w_s-\frac{\mu } {A\Big(\frac{w}{w_z}\Big)}w-\frac{1}{2}\si^2A\Big(\frac{w}{w_z}\Big)w+(\si^2+\mu)w\\
&>
-r_1 w-\frac{\mu}{a_1}w-\frac{1}{2}\si^2a_2w+(\si^2+\mu)w \\
&=
\si^2 w\Big( a_1-\frac{\mu}{\si^2 a_1}-\frac{1}{2} a_2+1\Big)\quad\rm{in}\quad {\cal N}.
\end{align*}
The inequality is strict because $A\Big(\frac{w}{w_z}\Big)$ cannot equal $a_1$ and $a_2$ simultaneously.
It follows
\[
\p_z\Big(\frac{w}{w_z}\Big)> 1-2 \Big( a_1-\frac{\mu}{\si^2 a_1}-\frac{1}{2} a_2+1\Big)=1-2 a_1+a_2+\frac{2 r_1}{\mu-r_1} \geq0\quad\rm{in}\quad {\cal N}.\]
This completes the proof.
\end{proof}

Now, we are ready to prove \propref{prop:C1}. Let $J=-w/w_z$, from the definition of $b(\cdot)$ and $l(\cdot)$ we have $$J(b(s),s)=a_2,\quad J(l(s),s)=a_1,\quad s\in (0,T].$$ When the conditions \eqref{CD1} and \eqref{CD2} hold, the above result shows $J_z(b(s),s)$ and $J_z(l(s),s)<0$. So it follows from the implicit function existence theorem that $b(\cdot)$, $l(\cdot)\in C^1((0,T])$.

\section{Solutions to the HJB Equation \eqref{V_pb} and Problem \eqref{value}}\label{sec:V_solu}

We are now ready to construct a classical solution to the PDE \eqref{V_pb} from the function $w$ given in \thmref{theo:w} and deduce the optimal portfolio to the problem \eqref{value}.

First, we rewritten the PDEs \eqref{u_eq} and \eqref{v_eq} of $u$ and $v$ in compact forms as follows.
\begin{align}\label{u_pb}
\begin{cases}
-u_t-{\cal J}u=0\quad\hbox{in}\quad (0,+\infty)\times [0,T),\\[3mm]
u(y,T)=\frac{1}{2}y-d,\quad y>0,
\end{cases}
\end{align}
and
\begin{align}\label{v_pb}
\begin{cases}
-v_t-{\cal H}v=0
\quad\hbox{in}\quad (0,+\infty)\times [0,T),\\[3mm]
v(0,t)=0, \qquad 0\leq t<T,\\[3mm]
v(y,T)=-\frac{1}{4}y^2+d y,\quad y>0,
\end{cases}
\end{align}
where
\[
{\cal J}u:=\frac{1}{2}\si^2 A^2\Big(\frac{u}{yu_y}\Big)y^2u_{yy}-\Big(\mu-\si^2A^2\Big(\frac{u}{yu_y}\Big)-\si^2A\Big(\frac{u}{yu_y}\Big) \Big)yu_y-\Big(\mu-\si^2A\Big(\frac{u}{yu_y}\Big)\Big)u,\]
and
\[
{\cal H}v:=\frac{1}{2}\si^2 A^2\Big(\frac{v_y}{yv_{yy}}\Big)y^2v_{yy}-\Big(\mu-\si^2A\Big(\frac{v_y}{yv_{yy}}\Big)\Big)yv_y.\]

\begin{lemma}\label{lem:u}
Let $w$ be given in Theorem \ref{theo:w} and let $$u(y,t)=w(\ln y, T-t).$$ Then
$u\in C^{2+\al,1+\frac{\al}{2}}\big((0,+\infty)\times[0,T]\big)$ is a solution to the PDE \eqref{u_pb} such that
\begin{align} \label{u_b}
\frac{1}{2} e^{\theta_2 (T-t)} y-e^{-r_1 (T-t)} d\leq & u \leq \frac{1}{2} e^{\theta_1 (T-t)} y-e^{-r_2 (T-t)} d,\\\label{uy_b}
\frac{1}{2}e^{-\kappa (T-t)} \leq & u_y\leq \frac{1}{2}e^{k (T-t)},\end{align}
in $(0,+\infty)\times[0,T]$.
\end{lemma}
This result can be easily verified, so we omit its proof.

Furthermore, we have
\begin{lemma}\label{lem:u_lim}
For any $t\in [0,T]$,
\begin{align}\label{u_lim}
\lim\limits_{y\rightarrow 0+}u=-e^{-r_1 (T-t)} d,\quad
\lim\limits_{y\rightarrow+\infty}u=+\infty, \quad
\lim\limits_{y\rightarrow 0+}yu=0, \quad
\lim\limits_{y\rightarrow 0+}y^2u_y=0.\end{align}
\end{lemma}
\begin{proof}
The second and third limits can be derived from \eqref{u_b}, the fourth limit is due to \eqref{uy_b}. It is left to prove the first limit. Thanks to the estimate \eqref{l_lb}, there exists $$z_0\in\Big(-\infty, \inf\limits_{s\in [0,T]}l(s)\Big)$$ such that $A(w/w_z)=a_1$ in 
$${\cal D}:=(-\infty,z_0]\times[0,T].$$ Thus,
\bee
w_s-\frac{1}{2}\si^2 a_1^2w_{zz}+\big(\mu-\frac{1}{2}\si^2a_1^2-\si^2a_1\big)w_z+(\mu-\si^2a_1)w=0 \quad \hbox{in}\quad {\cal D}.\eee
Let $$M:=\max\big\{\frac{1}{2}e^{z_0},\max\limits_{s\in [0,T]} w(z_0,s)+e^{-r_1 s}d\big\}$$ and denote $$\Psi(z,s):=M e^{|\theta_1| s} e^{z-z_0}-e^{-r_1 s} d.$$ Then
\bee
&& \Psi_s-\frac{1}{2}\si^2 a_1^2 \Psi_{zz}+\Big(-\frac{1}{2}\si^2 a_1^2-\si^2 a_1+\mu\Big)\Psi_z+\Big(\mu-\si^2 a_1 \Big)\Psi\\
&=&
M e^{|\theta_1| s} e^{z-z_0}\Big(|\theta_1|-\frac{1}{2}\si^2 a_1^2+\Big(-\frac{1}{2}\si^2 a_1^2-\si^2 a_1+\mu\Big)+\Big(\mu-\si^2 a_1 \Big)\Big)\\
&&+e^{-r_1 s} d (r_1-(\mu-\si^2 a_1 ) )\\
&\geq&0,\eee
by recalling the definitions of $\theta_1$ and $a_1$.
Moreover,
\[
\begin{cases}
\Psi(z,0)=Me^{z-z_0}-d\geq \frac{1}{2}e^{z}-d=w(z,0), \quad z\leq z_0, \\[3mm]
\Psi(z_0,s)\geq M-e^{-r_1 s}d\geq w(z_0,s), \quad s\in [0,T],
\end{cases}
\]
so, by the comparison principle, we get $\Psi\geq w$ in ${\cal D}$. Together with the first inequality in \eqref{w_b}, we have $\lim\limits_{z\rightarrow-\infty} w=-e^{-r_1s} d$, which implies $\lim\limits_{y\rightarrow 0+}u=-e^{-r_1 (T-t)} d$.
\end{proof}
It follows from \eqref{uy_b} and \eqref{u_lim} that $$ u_y>0,\quad\lim\limits_{y\rightarrow 0+}u=-e^{-r_1 (T-t)} d,\quad \lim\limits_{y\rightarrow+\infty}u=+\infty, $$ so $-u$ is one-to-one mapping $(0,+\infty) $ to $(-\infty, e^{-r_1 (T-t)} d)$ for each $t\in[0,T)$.

\begin{lemma}\label{lem:v}
Let $u$ be given in \lemref{lem:u}. Define
\begin{align*}
v(y,t):=-\int_0^yu(\xi,t){\rm d} \xi,\quad (y,t)\in (0,+\infty)\times[0,T].
\end{align*}
Then $v\in C^{3,2}\big((0,+\infty)\times[0,T]\big)$ is a solution to the PDE \eqref{v_pb} such that
\begin{align} \label{vy_b}
-\frac{1}{2} e^{\theta_1 (T-t)} y+e^{-r_2 (T-t)}d\leq & v_y \leq-\frac{1}{2} e^{\theta_2 (T-t)} y+e^{-r_1 (T-t)} d,\\\label{vyy_b}
-\frac{1}{2}e^{k (T-t)}\leq & v_{yy}\leq-\frac{1}{2}e^{-\kappa (T-t)},\end{align}
in $(0,+\infty)\times[0,T]$.
Moreover, for any $t\in [0,T]$,
\begin{align}\label{v_lim}
\lim\limits_{y\rightarrow 0+}v_y=e^{-r_1 (T-t)} d,\quad
\lim\limits_{y\rightarrow+\infty}v_y=-\infty, \quad
\lim\limits_{y\rightarrow 0+}y v_y=0\quad
\lim\limits_{y\rightarrow 0+}y^2 v_{yy}=0.\end{align}
\end{lemma}
\begin{proof}
Clearly $v_y=-u$, so \eqref{vy_b}, \eqref{vyy_b} and \eqref{v_lim} are the direct consequences of \eqref{u_b}, \eqref{uy_b} and \eqref{u_lim} respectively.
Since $u\in C^{2+\al,1+\frac{\al}{2}}\big((0,+\infty)\times[0,T]\big)$,
we have $v$, $v_y\in C^{2+\al,1+\frac{\al}{2}}\big((0,+\infty)\times[0,T]\big)$.
Moreover, it is easy to check that
$$\p_y(-v_t-{\cal H} v)=u_t+{\cal J}u=0$$ and $$(-v_t-{\cal H}v)(0,t)=0,$$
so
\begin{align*}
(-v_t-{\cal H} v)(y,t)
=(-v_t-{\cal H} v)(0,t)+\int_0^y\p_y(-v_t-{\cal H} v)(\xi,t){\rm d}\xi
=0.\end{align*}
Therefore, $v$ is a solution to the PDE \eqref{v_pb}. As a consequence,
\[
-v_{tt}-\frac{1}{2}\si^2 A^2\Big(\frac{v_y}{yv_{yy}}\Big)y^2v_{tyy}+\Big(\mu-\si^2A\Big(\frac{v_y}{yv_{yy}}\Big)\Big)yv_{ty}=\p_t(-v_t-{\cal H} v)=0.\]
Using the Schauder interior estimation (see \cite{Li96} Theorem 4.9), we get $v_t\in C^{2+\al,1+\frac{\al}{2}}\big((0,+\infty)\times[0,T]\big)$,
therefore, $v\in C^{3,2}\big((0,+\infty)\times[0,T]\big)$.
\end{proof}

\subsection{Proof of \thmref{theo:V}}\label{proofv}
Now we are ready to prove \thmref{theo:V}.
Let $v$ be given in \lemref{lem:v}. Define
\begin{align}\label{V_duel}
V (x,t):=\sup\limits_{y>0}\;(v(y,t)-xy),\quad x<e^{r_1 (T-t)}d,\quad t\in [0,T].
\end{align} 
We come to prove that the above $V$ satisfies the requirements of \thmref{theo:V}.

For each $t\in[0,T]$, the estimates \eqref{vyy_b} and \eqref{v_lim} imply $v_y(\cdot,t)$ is strictly decreasing and maps $(0,\infty)$ to $(-\infty,e^{-r_1 (T-t)} d)$, so
\begin{align*}
J(x,t):=\argmax\limits_{y>0}\;(v(y,t)-xy)=
(v_y(\cdot,t))^{-1}(x)>0, \end{align*}
and
\begin{align}\label{V_J}
V (x,t)=v(J(x,t),t)-xJ(x,t),\quad x<e^{r_1 (T-t)}d,\quad t\in [0,T].
\end{align}
Also the function $J(x,t)\in C(\TQ)$ and is strictly increasing w.r.t. $x$. Therefore,
\begin{align*}
V_x(x,t) &=v_y(J(x,t),t)J_x(x,t)-xJ_x(x,t)-J(x,t)=-J(x,t)<0,\\
V_{xx}(x,t)&=-J_x(x,t)=-\partial_x[(v_y(\cdot,t))^{-1}(x)]=\frac{-1}{v_{yy}(J(x,t),t)}>0,\\
V_t(x,t)&=v_y(J(x,t),t)J_t(x,t)+v_t(J(x,t),t)-xJ_t(x,t)=v_t(J(x,t),t).
\end{align*}
As $v\in C^{3,2}\big((0,+\infty)\times[0,T]\big)$, we get $$ V\in C^{3,2}\big(\overline{\TQ}\setminus\{x=e^{-r_1 (T-t) d}\} \big).$$ Since $v$ is a solution to the PDE \eqref{v_pb}, which is equivalent to \eqref{v_eq}, one can check that $ V$ satisfies the PDE in \eqref{Three_eq}. This together with $V_x<0$ and $V_{xx}>0$ shown above implies the PDE in \eqref{V_pb}.

From \eqref{v_lim} we know for any $t\in[0,T]$,
\begin{align}\label{J_lim}
\lim\limits_{x\rightarrow e^{-r_1 (T-t)} d-}J(x,t)=0,\quad
\lim\limits_{x\rightarrow-\infty}J(x,t)=+\infty.\end{align}
So \eqref{V_lim} holds.
Moreover, \eqref{J_lim} and \eqref{V_J} imply $$V(e^{-r_1 (T-t)} d-,t)=v(0+,t)=0,$$ so the boundary condition in \eqref{V_pb} holds.

Now, we verify the terminal condition. Thanks to \eqref{vy_b} and $v(0,t)=0$, we have
$$
-\frac{1}{4} e^{\theta_1 (T-t)} y^2+e^{-r_2 (T-t)}d y\leq v \leq-\frac{1}{4} e^{\theta_2 (T-t)} y^2+e^{-r_1 (T-t)} d y,
$$
and consequently,
\begin{align}
V (x,t)&=\sup\limits_{y>0}\;(v(y,t)-xy)\nonumber\\
&\geq \sup\limits_{y>0}\Big(-\frac{1}{4} e^{\theta_1 (T-t)} y^2+e^{-r_2 (T-t)}d y-xy\Big)\nonumber\\
&=e^{-\theta_1 (T-t)}(e^{-r_2 (T-t)}d-x)^2,\label{V_B1}
\end{align}
and
\begin{align}\label{V_B2}
V (x,t)
&\leq \sup\limits_{y>0}\Big(-\frac{1}{4} e^{\theta_2 (T-t)} y^2+e^{-r_2 (T-t)}d y-xy\Big)\nonumber\\
&= e^{-\theta_2 (T-t)}(e^{-r_1 (T-t)}d-x)^2.
\end{align}
Letting $t\rightarrow T$ in the above two inequalities, it follows that $V$ satisfies the terminal condition in \eqref{V_pb}.
This completes the proof of \thmref{theo:V}.

\subsection{Proof of \propref{theo:free_boundary}}\label{sec:freeboundary}
This is the consequence of \lemref{lem:B} and \lemref{lem:L}.
Thanks to \eqref{Vx_b} and \eqref{Vxx_b}, we see $(0,t)\in\fB$, so $B(t)>0$.

\subsection{Proof of Theorem \ref{veri}}\label{sec:veri}

In this section, we prove Theorem \ref{veri}.

Suppose $V$ is the solution to \eqref{V_pb} given in Theorem \ref{theo:V}. Fix any $(x,t)\in \TQ$ and any admissible portfolio $\pi\in\Pi_t$, let $X$ be the unique strong solution to \eqref{X_eq}. We set
$${\tau_n=\inf\Big\{s\geq t \;\Big|\; |V_x(X_s,s)|+\int_t^s|\pi_u|^2\d u\geq n\Big\}.}$$ Then $$s\mapsto\int_t^{s\wedge \tau_n\wedge T}V_x(X_u,u)\si\pi_u \d W_u$$ is a martingale, whose mean is 0.
Therefore, by It\^o's formula and the HJB equation \eqref{V_pb},
\begin{align}
&\quad\;\E[V(X_{T\wedge \tau_n},T\wedge \tau_n)\mid X_t=x ]\nonumber\\
&=V(x,t)+\E\Big[\int_t^{T\wedge \tau_n}\Big(V_t+\frac{1}{2}\si^2\pi_s^2V_{xx}+\Big((r_1\chi_{X_s>\pi_s}+r_2\chi_{X_s<\pi_s})(X_s-\pi_s)\nonumber \\
&\qquad\qquad\qquad\qquad\qquad\quad\;+\mu\pi_s\Big)V_x\Big)(X_s,s)\d s+\int_t^{T\wedge \tau_n}V_x(X_s,s)\si\pi_s \d W_s\;\Big|\; X_t=x \Big]\nonumber\\
&\geq V(x,t).\label{veq1}
\end{align}
Using the estimates \eqref{V_B1} and \eqref{V_B2}, we have
\[
0\leq V(X_{T\wedge \tau_n},T\wedge \tau_n)\leq C\big(1+\sup\limits_{s\in[t,T]}|X_s|^2\big).
\]
By the standard estimate for SDE, the right hand side is integrable, so we can apply the dominated convergence theorem to $\E[V(X_{T\wedge \tau_n},T\wedge \tau_n)\mid X_t=x ]$, and obtain
\begin{align}
V(x,t) &\leq\E[\lim\limits_{n\rt \infty}V(X_{T\wedge \tau_n},T\wedge \tau_n)\mid X_t=x ]\nonumber\\
&=\E[V(X_{T},T)\mid X_t=x ]\nonumber\\
&=\E[(X_T-d)^2 \mid X_t=x ].
\label{veq2}
\end{align}
Therefore, we have
\begin{align}\label{lowerv}
V(x,t)\leq\inf\limits_{\pi\in\Pi_t}\E[(X_T-d)^2 \mid X_t=x ].
\end{align}

To show the reverse inequality, define a feedback control
\[
\pi^*(x,t)=
\begin{cases}
-a_1\frac{V_{x}(x,t)}{V_{xx}(x,t)}, &-a_1\frac{V_{x}(x,t)}{V_{xx}(x,t)}<x,\\[3mm]
x, &-a_2\frac{V_{x}(x,t)}{V_{xx}(x,t)}\leq x\leq-a_1\frac{V_{x}(x,t)}{V_{xx}(x,t)},\\[3mm]
-a_2\frac{V_{x}(x,t)}{V_{xx}(x,t)}, &-a_2\frac{V_{x}(x,t)}{V_{xx}(x,t)}>x.
\end{cases}
\]
Clearly \eqref{vy_b} and \eqref{vyy_b} imply
$|y v_{yy}|\leq C(1+|v_y|)$ for some constant $C$ independent of $t$, which is equivalent to $$\Big|\frac{V_{x}(x,t)}{V_{xx}(x,t)}\Big|\leq C(1+|x|).$$ This indicates $\pi^*(x,t)$ is linear growth in $x$ uniformly for $t\in[0,T]$. Moreover, because $V\in C^{3,2}$, $\pi^*(x,t)$ is locally Lipschitz continuous. By Mao \cite[Theorem 3.4, p.56]{M08}, there exists a unique strong solution $X^*$ to the following SDE:
\begin{equation}\label{X*_eq}
\left\{\begin{array}{rl}
{\rm d}X^*_s&=\big[\big(r_1\chi_{X^*_s>\pi^*(X^*_s,s)}+r_2\chi_{X^*_s<\pi^*(X^*_s,s)}\big)(X^*_s-\pi^*(X^*_s,s))\\[2mm]
&\quad\quad\quad\quad\quad\quad\quad\quad\quad\quad\quad
+\mu\pi^*(X^*_s,s)\big]{\rm d}s+\sigma\pi^*(X^*_s,s){\rm d}W_s,\quad s\in [t,T], \\[2mm]
X^*_t&=x.
\end{array}\right.
\end{equation}
Furthermore, as $\pi^*(x,t)$ is linear growth in $x$, by Mao \cite[Lemma 3.2, p.51]{M08}, we obtain from \eqref{X*_eq} that
\[\E\Big[ \sup\limits_{s\in[t,T]}|X^*_s|^2\Big]<\infty,\]
which further implies $\hat\pi_s:=\pi^*(X^*_s,s)$ is an admissible control in $\Pi_t$ by the linear growth property of $\pi^*(x,t)$ in $x$. Repeat the preceding argument with the control $\hat\pi$, then the inequalities in \eqref{veq1} and \eqref{veq2} become equations, giving $$V(x,t)=\E [(X^*_T-d)^2| X_t=x ].$$
Compared to \eqref{lowerv}, we conclude that $\pi^*$ is an optimal feedback control to the problem \eqref{value}, and $V$ is the value function.

\section{Concluding Remarks}\label{sec:cr}

In this paper, we solved Markowitz’s mean-variance portfolio selection problem in a continuous-time Black-Scholes market with different borrowing and saving rates by PDE methods. A feedback optimal portfolio is provided. Efficiently numerical schemes can be easily developed to calculate it. Different from many existing papers, the optimality of the portfolio is proved by a verification argument, where the smoothness of the value function plays an important role. It is of great interests to extend our model to the case with jumps in stock price. This may lead to some new financial insights.

Clearly, \thmref{veri} implies the function $V$ given in \thmref{theo:V} is unique. As a consequence, we have that the functions $w$ in \thmref{theo:w}, $u$ in \lemref{lem:u} and $v$ in \lemref{lem:v} are unique as well. The above uniqueness can be proved by pure PDE argument as well. We leave this to the interested readers. 

This paper used PDE method to solve the portfolio selection problem. This approach does not work if the system is not Markovian or $d$ is stochastic in general. So stochastic control theory for piecewise linear quadratic problems is called for. Of course, it is of great importance to develop such theories, and also far beyonds the scope of this paper. But we hope our method can inspire the readers to develop such theories.

\newpage
\begin{appendices}

\section*{Appendix: Proof of Theorem \ref{theo:w}}\label{sec:w_exist} 
In this section, we prove Theorem \ref{theo:w} by approximation method.

Firstly, for each fixed $0<\ep<1$, define a continuous function
\[
\Gamma_\ep(\xi,\eta):=A\Big(\frac{\xi}{\eta+\ep}\Big),\quad (\xi,\eta)\in (-\infty,+\infty)\times[0,+\infty).\]
Note that
\[
\p_\xi \Gamma_\ep(\xi,\eta)=A^{\prime}\Big(\frac{\xi}{\eta+\ep}\Big) \frac{1}{\eta+\ep}=
\left\{
\begin{array}{ll}
-\frac{1}{\eta+\ep}\in[-\frac{1}{\ep},0),&\quad \mbox{if}\;a_2<-\frac{\xi}{\eta+\ep}<a_1,\bigskip\\
0,&\quad \mbox{if}\;-\frac{\xi}{\eta+\ep}>a_1 \;\hbox{or}\; -\frac{\xi}{\eta+\ep}<a_2,
\end{array}
\right.\]
and
\[
\p_\eta \Gamma_\ep(\xi,\eta)=A^{\prime}\Big(\frac{\xi}{\eta+\ep}\Big) \frac{-\xi}{(\eta+\ep)^2}=
\left\{
\begin{array}{ll}
\frac{\xi}{\eta+\ep}\frac{1}{\eta+\ep}\in [-\frac{a_1}{\ep},0),&\quad \mbox{if}\; a_2<-\frac{\xi}{\eta+\ep}<a_1,\bigskip\\
0,&\quad \mbox{if}\; -\frac{\xi}{\eta+\ep}>a_1\;\hbox{or}\;-\frac{\xi}{\eta+\ep}<a_2,
\end{array}
\right.\]
so the function $\Gamma_\ep(\cdot,\cdot)$ is Lipschitz continuous in $(-\infty,+\infty)\times[0,+\infty) $.
Moreover, for each fixed $c>0$, $\p_\xi\Gamma_\ep(\xi,\eta)$ and $\p_\eta \Gamma_\ep(\xi,\eta)$ are uniformly bounded for all $(\xi,\eta, \ep)\in (-\infty,+\infty)\times[c,+\infty)\times [0,1]$.

Now, consider an approximation equation in a bounded domain $Q_T^N:=(-N,N)\times[0,T]$,
\begin{align}\label{wNN_pb}
\left\{
\begin{array}{ll}
w^{\ep,N}_s-\frac{1}{2}\si^2A^2\Big(\frac{w^{\ep,N}}{|w^{\ep,N}_z|+\ep}\Big)w^{\ep,N}_{zz}+\Big(\mu-\frac{1}{2}\si^2A^2\Big(\frac{w^{\ep,N}}{|w^{\ep,N}_z|+\ep}\Big)-\si^2A\Big(\frac{w^{\ep,N}}{|w^{\ep,N}_z|+\ep}\Big)\Big)w^{\ep,N}_z\\[5mm]
\qquad\qquad\qquad\qquad\qquad\qquad\qquad\qquad\quad\quad\quad
+\Big(\mu-\si^2A\Big(\frac{w^{\ep,N}}{|w^{\ep,N}_z|+\ep}\Big)\Big)w^{\ep,N}=0 \quad \hbox{in} \quad Q_T^N,\\[5mm]
(w^{\ep,N}-w^{\ep,N}_z)(-N,s)=-e^{-r_2 s }d,\quad w^{\ep,N}_z(N,s)=\frac{1}{2} e^{\theta_1 s } e^N,\quad s\in [0,T],\\[5mm]
w^{\ep,N}(z,0)=\frac{1}{2}e^z-d,\quad-N<z<N,
\end{array}
\right.\end{align}
The Leray-Schauder fixed point theorem (see \cite{Ev16} Theorem 4, p.541) and embedding theorem (see \cite{Li96} Theorem 6.8) imply the existence of $C^{1+\al,\frac{1+\al}{2}}\big(\overline{Q_T^N}\big)$ (for some $\al\in(0,1)$) solution to the problem \eqref{wNN_pb}.
Moreover, the Schauder estimation (see \cite{Li96} Theorem 4.23) implies 
$$w^{\ep,N}\in C^{2+\al,1+\frac{\al}{2}}\big(\overline{Q_T^N}\big).$$

In the proceeding proof, we will frequently use the following fact without claim: 

$$0<a_2\leq A(\xi)\leq a_1,\quad |A'(\xi)| \leq 1,\quad a_2\leq |A'(\xi)\xi| \leq a_1,\quad a_2^2\leq |A'(\xi)\xi^2| \leq a_1^2.$$ 

We first establish the estimates
\begin{align}\label{wNN_b}
\frac{1}{2} e^{\theta_2 s } e^z-e^{-r_1 s } d\leq w^{\ep,N} \leq \frac{1}{2} e^{\theta_1 s } e^z-e^{-r_2 s }d.\end{align}
Denote $$\psi(z,s)=\frac{1}{2} e^{\theta_2 s } e^z-e^{-r_1 s } d,\quad A(\cdot\cdot)=A\Big(\frac{w^{\ep,N}}{|w^{\ep,N}_z|+\ep}\Big).$$ Using the definitions of $\theta_2$, $a_1$ and $a_2$ as well as the bounds on $A$ and $A'$, we get
\bee
&& \psi_s-\frac{1}{2}\si^2A^2(\cdot\cdot)\psi_{zz}+\Big(\mu-\frac{1}{2}\si^2A^2(\cdot\cdot)-\si^2A(\cdot\cdot)\Big)\psi_z+\Big(\mu-\si^2A(\cdot\cdot)\Big)\psi\\
&=&
\frac{1}{2}e^{\theta_2 s } e^z\Big(\theta_2-\frac{1}{2}\si^2A^2(\cdot\cdot)+\Big(\mu-\frac{1}{2}\si^2A^2(\cdot\cdot)-\si^2A(\cdot\cdot)\Big)+\Big(\mu-\si^2A(\cdot\cdot)\Big)\Big)\\
&&
+e^{-r_1 s } d\Big(r_1-(\mu-\si^2A(\cdot\cdot))\Big)\\
&\leq&
\frac{1}{2}e^{\theta_2 s } e^z\Big(\theta_2-\si^2a_2^2-2\si^2a_2+2\mu\Big)
+e^{-r_1 s } d\Big(r_1-(\mu-\si^2a_1)\Big)\\
&=&0.\eee
Notice $\theta_1>\theta_2$, so
\[
\begin{cases}
\psi(z,0)=\frac{1}{2}e^z-d=w^{\ep,N}(z,0), &-N<z<N, \\[3mm]
(\psi-\psi_z)(-N,s)=-e^{-r_1 s } d\leq-e^{-r_2 s }d=(w^{\ep,N}-w^{\ep,N}_z)(-N,s),& s\in [0,T], \\[3mm]
\psi_z(N,s)=\frac{1}{2} e^{\theta_2 s } e^N\leq \frac{1}{2} e^{\theta_1 s } e^N=w^{\ep,N}_z(N,s), & s\in [0,T].
\end{cases}
\]
Applying the comparison principle for linear equations, the first inequality in \eqref{wNN_b} is established.

Similarly, let $$\Psi(z,s)=\frac{1}{2}e^{\theta_1 s }e^{z}-e^{-r_2 s }d.$$ Then by the definitions of $\theta_1$, $a_1$ and $a_2$,
\bee
&& \Psi_s-\frac{1}{2}\si^2A^2(\cdot\cdot)\Psi_{zz}+\Big(\mu-\frac{1}{2}\si^2A^2(\cdot\cdot)-\si^2A(\cdot\cdot)\Big)\Psi_z+\Big(\mu-\si^2A(\cdot\cdot)\Big)\Psi\\
&=&
\frac{1}{2}e^{\theta_1 s }e^{z}\Big(\theta_1-\frac{1}{2}\si^2A^2(\cdot\cdot)+\Big(\mu-\frac{1}{2}\si^2A^2(\cdot\cdot)-\si^2A(\cdot\cdot)\Big)+\Big(\mu-\si^2A(\cdot\cdot)\Big)\Big)\\
&&
+e^{-r_2 s }d\Big(r_2-(\mu-\si^2A(\cdot\cdot))\Big)\\
&\geq&
\frac{1}{2}e^{\theta_1 s }e^{z}\Big(\theta_1-\si^2a_1^2-2\si^2a_1+2\mu\Big)
+e^{-r_2 s }d\Big(r_2-(\mu-\si^2a_2)\Big)\\
&=&0.\eee
Moreover,
\[
\begin{cases}
\Psi(z,0)=\frac{1}{2}e^{z}-d=w^{\ep,N}(z,0), &-N<z<N, \\[3mm]
(\Psi-\Psi_z)(-N,s)=-e^{-r_2 s }d=(w^{\ep,N}-w^{\ep,N}_z)(-N,s),& s\in [0,T], \\[3mm]
\Psi_z(N,s)=\frac{1}{2} e^{\theta_1 s } e^N=w^{\ep,N}_z(N,s), & s\in [0,T],
\end{cases}
\]
by the comparison principle, the second inequality in \eqref{wNN_b} is established.


Due to the setting of boundary conditions, we cannot establish $w^{\ep,N}_z\geq \frac{1}{2}e^{-\kappa s} e^z$. Instead, we first prove
\be\label{wNz_lb1}
w^{\ep,N}_z\geq-e^{-\theta_3 s} d,
\ee
where $$\theta_3=\min\{\mu-\si^2 a_1(a_1+3),r_1\}.$$ Differentiating the equation in \eqref{wNN_pb} w.r.t. $z$ we have
\begin{multline*}
\p_s w^{\ep,N}_z-\frac{\si^2}{2}\p_z\Big(A^2(\cdot\cdot)\p_zw^{\ep,N}_z\Big)
+\Big(\mu-\frac{1}{2}\si^2A^2(\cdot\cdot)-\si^2A(\cdot\cdot)\Big)\p_zw^{\ep,N}_z
+\Big(\mu-\si^2A(\cdot\cdot)\Big)w^{\ep,N}_z\\
-\si^2A'(\cdot\cdot)\Big(\frac{w^{\ep,N}_z}{|w^{\ep,N}_z|+\ep}-\frac{w^{\ep,N}}{(|w^{\ep,N}_z|+\ep)^2}\sgn(w^{\ep,N}_z)w^{\ep,N}_{zz}\Big) \Big(A(\cdot\cdot)+1\Big)w^{\ep,N}_z\\
-\si^2A'(\cdot\cdot)\Big(\frac{w^{\ep,N}_z}{|w^{\ep,N}_z|+\ep}-\frac{w^{\ep,N}}{(|w^{\ep,N}_z|+\ep)^2}\sgn(w^{\ep,N}_z)w^{\ep,N}_{zz}\Big) w^{\ep,N}=0.
\end{multline*}
After reorganizing, we get an equation for $w^{\ep,N}_z$ in the divergence form:
\begin{multline}\label{wNNz_eq}
\p_s w^{\ep,N}_z-\frac{\si^2}{2}\p_z\Big(A^2(\cdot\cdot)\p_zw^{\ep,N}_z\Big)
+\Big(-\frac{1}{2}\si^2A^2(\cdot\cdot)-\si^2A(\cdot\cdot)+\mu\Big)\p_zw^{\ep,N}_z\\
+\Big(\mu-\si^2A(\cdot\cdot)\Big)w^{\ep,N}_z-\si^2A'(\cdot\cdot)\frac{w^{\ep,N}_z}{|w^{\ep,N}_z|+\ep} \Big(A(\cdot\cdot)+1\Big)w^{\ep,N}_z\\
+\si^2A'(\cdot\cdot)\Big(\frac{w^{\ep,N}}{|w^{\ep,N}_z|+\ep}\Big) \Big(\frac{w^{\ep,N}_z}{|w^{\ep,N}_z|+\ep}\Big) \Big(A(\cdot\cdot)+1\Big)\sgn(w^{\ep,N}_z)\p_zw^{\ep,N}_z\\
+\si^2A'(\cdot\cdot) \Big(\frac{w^{\ep,N}}{|w^{\ep,N}_z|+\ep}\Big)^2 \sgn(w^{\ep,N}_z)\p_zw^{\ep,N}_z-\si^2A'(\cdot\cdot) \frac{w^{\ep,N}}{|w^{\ep,N}_z|+\ep} w^{\ep,N}_z
=0.
\end{multline}
It is not hard to check that all the coefficients in \eqref{wNNz_eq} are bounded. Denote $\psi(z,s)=-e^{-\theta_3 s} d$, then
\begin{multline*}
\p_s \psi-\frac{\si^2}{2}\p_z\Big(A^2(\cdot\cdot)\p_z \psi\Big)
+\Big(\mu-\frac{1}{2}\si^2A^2(\cdot\cdot)-\si^2A(\cdot\cdot)\Big)\p_z \psi\\
+\Big(\mu-\si^2A(\cdot\cdot)\Big)\psi-\si^2A'(\cdot\cdot)\frac{w^{\ep,N}_z}{|w^{\ep,N}_z|+\ep} \Big(A(\cdot\cdot)+1\Big)\psi\\
+\si^2A'(\cdot\cdot)\Big(\frac{w^{\ep,N}}{|w^{\ep,N}_z|+\ep}\Big) \Big(\frac{w^{\ep,N}_z}{|w^{\ep,N}_z|+\ep}\Big) \Big(A(\cdot\cdot)+1\Big)\sgn(w^{\ep,N}_z)\p_z \psi\\
+\si^2A'(\cdot\cdot) \Big(\frac{w^{\ep,N}}{|w^{\ep,N}_z|+\ep}\Big)^2 \sgn(w^{\ep,N}_z)\p_z \psi-\si^2A'(\cdot\cdot) \frac{w^{\ep,N}}{|w^{\ep,N}_z|+\ep} \psi\\
=e^{-\theta_3 s}d\Big(\theta_3-\mu+\si^2A(\cdot\cdot)
+\si^2A'(\cdot\cdot)\frac{w^{\ep,N}_z}{|w^{\ep,N}_z|+\ep} \Big(A(\cdot\cdot)+1\Big)
+\si^2A'(\cdot\cdot) \frac{w^{\ep,N}}{|w^{\ep,N}_z|+\ep}
\Big)
\\
\leq e^{-\theta_3 s}d (\theta_3-\mu+\si^2 a_1+\si^2 a_1(a_1+1)+\si^2 a_1 )\leq 0,
\end{multline*}
thanks to the definition of $\theta_3$. Moreover,
\[
\begin{cases}
w^{\ep,N}_z(z,0)=\frac{1}{2}e^{z}\geq 0\geq \psi(z,0), \\[3mm]
w^{\ep,N}_z(-N,s)=w^{\ep,N}(-N,s)+e^{-r_2 s }d>-e^{-r_1 s} d\geq \psi(-N,s), \quad{(\rm by \;\eqref{wNN_b})}
\\[3mm]
w^{\ep,N}_z(N,s)=\frac{1}{2} e^{\theta_1 s } e^N\geq 0\geq \psi(N,s).
\end{cases}
\]
Using the comparison principle for divergence forms (see \cite{Li96} Corollary 6.16), we obtain $w^{\ep,N}_z\geq \psi$, giving \eqref{wNz_lb1}.

We next to prove
\be\label{wNNz_ub}
w^{\ep,N}_z\leq \frac{1}{2}e^{k s}e^z.
\ee
Denote $g^{\ep,N}(z,s)=e^{-z}w^{\ep,N}_z(z,s)$. According to \eqref{wNNz_eq}, we have
\begin{multline}\label{ezw}
\p_s g^{\ep,N}-\frac{\si^2}{2}\p_z\Big(A^2(\cdot\cdot)g^{\ep,N}_z\Big)
-\si^2A^2(\cdot\cdot)g^{\ep,N}_z
-\frac{\si^2}{2}A^2(\cdot\cdot)g^{\ep,N}
\\
-\si^2A(\cdot\cdot)A'(\cdot\cdot) \Big(\frac{w^{\ep,N}_z}{w^{\ep,N}_z+\ep}g-\Big(\frac{w^{\ep,N}}{|w^{\ep,N}_z|+\ep}\Big){\Big(\frac{w^{\ep,N}_z}{|w^{\ep,N}_z|+\ep}\Big)\sgn(w^{\ep,N}_z)}\big(g^{\ep,N}_z+g^{\ep,N} \big)\Big)
\\
+\Big(\mu-\frac{1}{2}\si^2A^2(\cdot\cdot)-\si^2A(\cdot\cdot)\Big)\big(g^{\ep,N}_z+g^{\ep,N} \big)\\
+\Big(\mu-\si^2A(\cdot\cdot)\Big)g^{\ep,N}
-\si^2A'(\cdot\cdot) \Big(A(\cdot\cdot)+1\Big)g^{\ep,N}\\
+\si^2A'(\cdot\cdot)\frac{w^{\ep,N}}{|w^{\ep,N}_z|+\ep} \Big(A(\cdot\cdot)+1\Big)\sgn(w^{\ep,N}_z)\big(g^{\ep,N}_z+g^{\ep,N} \big)\\
+\si^2A'(\cdot\cdot) \Big(\frac{w^{\ep,N}}{|w^{\ep,N}_z|+\ep}\Big)^2\sgn(w^{\ep,N}_z)\big(g^{\ep,N}_z+g^{\ep,N} \big)-\si^2A'(\cdot\cdot) \frac{w^{\ep,N}}{|w^{\ep,N}_z|+\ep} g^{\ep,N}
=0.
\end{multline}
On the other hand, denote $\Psi(z,s)=\frac{1}{2}e^{k s}$, then
\begin{multline*}
\p_s \Psi-\frac{\si^2}{2}\p_z\Big(A^2(\cdot\cdot)\Psi_z\Big)
-\si^2A^2(\cdot\cdot)\Psi_z-\frac{\si^2}{2}A^2(\cdot\cdot)\Psi\\
-\si^2A(\cdot\cdot)A'(\cdot\cdot) \Big(\frac{w^{\ep,N}_z}{w^{\ep,N}_z+\ep}\Psi-\Big(\frac{w^{\ep,N}}{|w^{\ep,N}_z|+\ep}\Big){\Big(\frac{w^{\ep,N}_z}{|w^{\ep,N}_z|+\ep}\Big)\sgn(w^{\ep,N}_z)}\big(\Psi_z+\Psi\big)\Big)\\
+\Big(-\frac{1}{2}\si^2A^2(\cdot\cdot)-\si^2A(\cdot\cdot)+\mu\Big)\big(\Psi_z+\Psi\big)\\
+\Big(\mu-\si^2A(\cdot\cdot)\Big)\Psi
-\si^2A'(\cdot\cdot) \Big(A(\cdot\cdot)+1\Big)\Psi\\
+\si^2A'(\cdot\cdot)\frac{w^{\ep,N}}{|w^{\ep,N}_z|+\ep} \Big(A(\cdot\cdot)+1\Big)\sgn(w^{\ep,N}_z)\big(\Psi_z+\Psi\big)\\
+\si^2A'(\cdot\cdot) \Big(\frac{w^{\ep,N}}{|w^{\ep,N}_z|+\ep}\Big)^2\sgn(w^{\ep,N}_z)\big(\Psi_z+\Psi\big)-\si^2A'(\cdot\cdot) \frac{w^{\ep,N}}{|w^{\ep,N}_z|+\ep} \Psi\\
\geq \frac{1}{2}e^{ks}{\Big(k-\frac{1}{2}\si^2 a_1^2-\si^2 a_1^2-\frac{1}{2}\si^2 a_1^2-\si^2 a_1(a_1+1)-\si^2 a_1^2-\si^2 a_1\Big)} \geq 0,
\end{multline*}
thanks to the definition of $k$. Notice $k\geq \theta_1$, so
\[
\begin{cases}
g^{\ep,N} (z,0)=\frac{1}{2}=\Psi(z,0),\\[3mm]
g^{\ep,N} (-N,s)=e^{N}(w^{\ep,N}+e^{-r_2 s}d)(-N,s)\leq \frac{1}{2} e^{\theta_1 s} \leq \frac{1}{2} e^{k s}=\Psi(-N,s),\quad{(\rm by \;\eqref{wNN_b})}\\[3mm]
g^{\ep,N} (N,s)=\frac{1}{2} e^{\theta_1 s } \leq \frac{1}{2} e^{k s}=\Psi(N,s).
\end{cases}
\]
Using the comparison principle for divergence forms, we obtain $g^{\ep,N}\leq \Psi$, proving \eqref{wNNz_ub}.

Thanks to \eqref{wNN_b}, \eqref{wNz_lb1} and \eqref{wNNz_ub},
for each $a<b$, when $N>\max\{|a|,|b|\}$, taking the $C^{\al,\frac{\al}{2}}$ interior estimate (see \cite{Li96} Theorem 6.33) to the equations in \eqref{wNN_pb} and \eqref{wNNz_eq} respectively, we obtain
\[
\Big|w^{\ep,N}\Big|_{C^{\al,\frac{\al}{2}}([a,b]\times[0,T])}, \quad \Big|w^{\ep,N}_z\Big|_{C^{\al,\frac{\al}{2}}([a,b]\times[0,T])}\;\leq C.\]
where $C$ is independent of $\ep$ and $N$. Since $\Gamma_\ep (\cdot,\cdot)$ is Lipschitz continuous in $(-\infty,+\infty)\times[0,+\infty) $, we have
\be\label{A_Ca}
\;\bigg|\;A\Big(\frac{w^{\ep,N}}{|w^{\ep,N}_z|+\ep}\Big)\;\bigg|\;_{C^{\al,\frac{\al}{2}}([a,b]\times[0,T])}\leq C_\ep
\ee
i.e. the coefficients in the equation of \eqref{wNN_pb} belong to $C^{\al,\frac{\al}{2}}([a,b]\times[0,T])$, so we can take the Schauder interior estimate to the equation in \eqref{wNN_pb} to get
\be\label{w_C2+a}
\Big|w^{\ep,N}\Big|_{C^{2+\al,1+\frac{\al}{2}}([a,b]\times[0,T])}\;\leq C_\ep.
\ee
where the above two $C_\ep$s are independent of $N$. Therefore, there exists $w^\ep\in C^{2+\al,1+\frac{\al}{2}}\big(\overline{Q_T}\big)$ such that, for any region $Q=(a,b)\times(0,T]\subset Q_T$, there exists a subsequence of $w^{\ep,N}$, which we still denote by $w^{\ep,N}$, such that $w^{\ep,N}\rightarrow w^\ep$ in $C^{2,1}(\overline{Q})$ when $N\rightarrow\infty$. So $w^\ep$ satisfies the initial problem
\begin{align}\label{wN_pb}
\left\{
\begin{array}{ll}
w^\ep_s-\frac{1}{2}\si^2A^2\big(\frac{w^\ep}{|w^\ep_z|+\ep}\big)w^\ep_{zz}+\Big(\mu-\frac{1}{2}\si^2A^2\big(\frac{w^\ep}{|w^\ep_z|+\ep}\big)-\si^2A\big(\frac{w^\ep}{|w^\ep_z|+\ep}\big)\Big)w^\ep_z\\[5mm]
\quad\quad\quad\quad\quad\quad\quad\quad\quad\quad\quad\quad\quad\quad\quad\quad\quad\quad
+\Big(\mu-\si^2A\big(\frac{w^\ep}{|w^\ep_z|+\ep}\big)\Big)w^\ep=0 \quad \hbox{in} \quad Q_T,\\[5mm]
w^\ep(z,0)=\frac{1}{2}e^z-d.
\end{array}
\right.\end{align}
The the exponential growth conditions on $w^\ep$ and $w^\ep_z$ come from the estimates \eqref{wNN_b}, \eqref{wNz_lb1} and \eqref{wNNz_ub}.

We now prove
\be\label{wNz_lb}
w^\ep_z\geq \frac{1}{2}e^{-\kappa s}e^{z}.
\ee
Denote $$g^\ep (z,s)=e^{-z}w^\ep_z(z,s),\quad A(\cdots)=A\big(\frac{w^\ep}{|w^\ep_z|+\ep}\big).$$
Letting $N\rightarrow\infty$ in \eqref{ezw}, we obtain
\begin{multline*}
\p_s g^\ep-\frac{\si^2}{2}\p_z\Big(A^2(\cdots) g^\ep_z\Big)
-\si^2A^2(\cdots) g^\ep_z
-\frac{\si^2}{2}A^2(\cdots) g^\ep \\
-\si^2A(\cdots)A'(\cdots) \Big({\frac{w^\ep_z}{|w^\ep_z|+\ep} g^\ep-\Big(\frac{w^{\ep}}{|w^{\ep}_z|+\ep}\Big)\Big(\frac{w^\ep_z}{|w^\ep_z|+\ep}\Big)\sgn(w^\ep_z)}\big( g^\ep_z+g^\ep \big)\Big)
\\
+\Big(\mu-\frac{1}{2}\si^2A^2(\cdots)-\si^2A(\cdots)\Big)\big( g^\ep_z+g^\ep \big)\\
+\Big(\mu-\si^2A(\cdots)\Big) g^\ep-\si^2A'(\cdots) \Big(A(\cdots)+1\Big) g^\ep \\
+\si^2A'(\cdots)\frac{w^\ep}{|w^\ep_z|+\ep} \Big(A(\cdots)+1\Big)\sgn(w^\ep_z)\big( g^\ep_z+g^\ep \big)\\
+\si^2A'(\cdots) \Big(\frac{w^\ep}{|w^\ep_z|+\ep}\Big)^2\sgn(w^\ep_z)\big( g^\ep_z+g^\ep\big)-\si^2A'(\cdots) \frac{w^\ep}{|w^\ep_z|+\ep} g^\ep=0.
\end{multline*}
On the other hand, denote $\Psi(z,s)=\frac{1}{2}e^{-\kappa s}$, we have
\begin{multline*}
\p_s \Psi-\frac{\si^2}{2}\p_z\Big(A^2(\cdots)\Psi_z\Big)
-\si^2A^2(\cdots)\Psi_z
-\frac{\si^2}{2}A^2(\cdots)\Psi\\
-\si^2A(\cdots)A'(\cdots)
\Big({\frac{w^\ep_z}{|w^\ep_z|+\ep} \Psi-\Big(\frac{w^{\ep}}{|w^{\ep}_z|+\ep}\Big)\Big(\frac{w^\ep_z}{|w^\ep_z|+\ep}\Big)\sgn(w^\ep_z)}\big(\Psi_z+\Psi\big)\Big)
\\
+\Big(-\frac{1}{2}\si^2A^2(\cdots)-\si^2A(\cdots)+\mu\Big)\big(\Psi_z+\Psi\big)\\
+\Big(\mu-\si^2A(\cdots)\Big)\Psi-\si^2A'(\cdots) \Big(A(\cdots)+1\Big)\Psi\\
+\si^2A'(\cdots)\frac{w^\ep}{|w^\ep_z|+\ep} \Big(A(\cdots)+1\Big)\sgn(w^\ep_z)\big(\Psi_z+\Psi\big)\\
+\si^2A'(\cdots) \Big(\frac{w^\ep}{|w^\ep_z|+\ep}\Big)^2\sgn(w^\ep_z)\big(\Psi_z+\Psi\big)-\si^2A'(\cdots) \frac{w^\ep}{|w^\ep_z|+\ep} \Psi\\
\leq \frac{1}{2}e^{-\kappa s}{\Big(-\kappa+\si^2 a_1(1+a_1)+\mu+\mu+\si^2(a_1+1)+\si^2a_1(a_1+1)+\si^2a_1^2+\si^2a_1\Big)}
=0,
\end{multline*}
thanks to the definition of $\kappa$. Moreover, $g(z,0)=\frac{1}{2}=\Psi(z,0)$. By the comparison principle we have $g\geq \Psi$, hence, \eqref{wNz_lb} is proved.

Thanks to \eqref{wNNz_ub} and \eqref{wNz_lb}, $w^\ep_z$ has positive lower and upper bounds which are independent of $\ep$ in any bounded region, noting that the bounds of $|\p_\eta \Gamma_\ep(\xi,\eta)|$ and $|\p_\eta \Gamma_\ep(\xi,\eta)|$ are independent of $\ep$ when $\eta$ has a positive lower bound,
so the constants $C_\ep$s in the estimates \eqref{A_Ca} and \eqref{w_C2+a} are independent of $\ep$. Let $\ep\rightarrow 0$ in \eqref{wN_pb}, we obtain a limit $w$ that satisfies \eqref{w_pb}. Moreover, \eqref{w_b} and \eqref{wz_b} are the direct consequences of \eqref{wNN_b}, \eqref{wNNz_ub}, \eqref{wNz_lb}.

\end{appendices}

\newpage

\bibliographystyle{plainnat}

\bibliographystyle{siam}
\renewcommand{\baselinestretch}{1.2}

\end{document}